\documentclass[11pt]{article}
\usepackage{natbib}
\usepackage{amsmath}
\usepackage{amsthm}
\usepackage{amssymb}
\usepackage{amsthm}
\usepackage{graphicx}
\usepackage{sidecap}
\usepackage{subcaption}
\usepackage{caption}
\usepackage{comment}
\usepackage{float}
\usepackage{setspace}
\usepackage{fullpage}
\usepackage{pdfpages}
\usepackage{import}
\usepackage{rotating}
\usepackage{hyperref}
\usepackage{cleveref} % clever referencing package
\usepackage{minitoc}
\newtheorem*{thm*}{Theorem}
\newtheorem{prop}{Proposition}
\newtheorem*{prop*}{Proposition} 
\newtheorem{cor}{Corollary}     
\newtheorem{lem}{Lemma}
\newtheorem*{lem*}{Lemma}

\newtheorem{rem}{Remark} 
\newtheorem*{rem*}{Remark}

\newcommand{\B}{\mathcal{B}}
\renewcommand{\d}{\mathrm{d}}
\usepackage{lineno}
\title{Inside the West Wing: Lobbying as a contest} 
\author{Alastair Langtry\footnote{University of Cambridge.  Email: \emph{atl27@cam.ac.uk.} I am very grateful to Toke Aidt, Matt Elliott, Aytek Erdil, Lukas Freund, Sanjeev Goyal, Max Langtry, Su Min Lee, Julia Shvets, and Sarah Taylor, and to participants at the Cambridge theory and political economy workshops, for helpful comments and discussion, and to Eric Snowberg and two anonymous referees for many insightful comments.
This work was supported by the Economic and Social Research Council [award reference ES/P000738/1]. Declaration of interests: none. Any remaining errors are the sole responsibility of the author.}}

\date{\today \\ \vspace{7mm}}

\begin{document}

\maketitle

\begin{abstract}
When a government makes many different policy decisions, lobbying can be viewed as a contest between the government and many different special interest groups. The government fights lobbying by interest groups with its own political capital. In this world, we find that a government wants to `sell protection' -- give favourable treatment in exchange for contributions -- to certain interest groups. It does this in order to build its own `war chest' of political capital, which improves its position in fights with other interest groups. And it does so until it wins all remaining contests with certainty. This stands in contrast to existing models that often view lobbying as driven by information or agency problems.
\end{abstract}

\newpage
Governments make many different policy decisions. In doing so, they often face pressures from different special interest groups who each care about one particular policy. These groups lobby for their special interest, which is often not aligned with overall welfare. Faced with many groups lobbying for many different policies, a government may not have the resources to get its way on every policy decision. So in practice, a government may `pick its battles' -- specifically, seek to do deals with some special interest groups to avoid a fight. However, existing models of lobbying struggle to engage with this intuition because they do not allow the government to do deals to avoid fights.  

We assume the government has limited resources. This gives rise to the possibility that some special interest groups get their way. Crucially, we also allow the government to cede a policy decision to a special interest group, and for special interest groups to make transfers to the government. Our key contribution is to show that this happens, and to explain why: in equilibrium, the government will not try to fight all the special interest groups but will pick its battles. Specifically, it either engages in a contest and wins with certainty, or it does a deal where it takes a contribution in exchange for ceding the contest to the special interest group.
Note that there are three key components of a deal in our model: (1) a special interest group gives some resources, (2) which the government uses either directly or indirectly to win other contests, (3) in exchange for ceding the contest the special interest group cares about.

As a motivating example, consider a deal done between the Obama administration and the pharmaceutical lobby to secure passage of the Affordable Care Act (`Obamacare'). In it, the Obama administration dropped provisions allowing for the re-importation of prescription medicines at lower prices \citep{NYT2012}, and the pharmaceutical industry agreed to provide \$80 billion to help fund expanded insurance coverage \citep{NYT2012} and spent millions of dollars on pro-Obamacare advertising campaigns \citep{RollCall2009}.
In a complex bill like Obamacare, with many provisions each affecting different interest groups, we view each provision as its own policy. So the Obama administration ceded a contest with the pharmaceutical industry in order to win others (and secure passage of Obamacare in some form).

But to tie this properly to our intuition, it must be that: (1) the Obama administration took (or was able to control) some specific resources from the pharmaceutical lobby and (2) either used them directly to help win different legislative battles, or used them to free up other resources that could then be used to win legislative battles. Finally, (3) this must be in exchange for the Obama administration ceding some contest with the pharmaceutical lobby. While we might expect that these deals between special interest groups and the government are typically kept behind closed doors, a leaked cache of private emails from the Obama administration provide direct evidence regarding the advertising spending. 

First, they show that \emph{PhRMA} (the pharmaceutical industry association) was willing to spend \$150 million on pro-Obamacare advertising through two groups it created, called \emph{Health Economy Now} and \emph{Americans for Stable Quality Care} \citep{memo_8_June}. So in this case, the `contribution' the special interest group made was in the form of advertising spending. They also show that this spending was directed and controlled by the Obama administration. This was through meetings and emails between Obama administration officials and representatives of the pharmaceutical lobby to agree the content and placement of the adverts. One leaked email puts this particularly bluntly: \emph{``Rahm} [Emmanuel, Chief of Staff to President Barack Obama] \emph{asked for Harry and Louise ads thru third party. We've already contacted the agent''} (email from Bryant Hall, a lobbyist for \emph{PhRMA} to Jeffrey Forbes, a political strategist, dated 7 July 2009). `Harry and Louise' were a fictional couple from a series of adverts attacking then-President Bill Clinton's proposed healthcare reforms in 1993 and 1994 \citep{harry_and_louise}. These ads were seen as playing a critical role in the failure of the Clinton healthcare reforms  \citep{west1996harry, brodie2001commentary}, giving Harry and Louise a special connection to healthcare reform. The adverts aired one week after the email \citep{memo_8_June}.

Second, they suggest this spending was aimed at securing passage of Obamacare. While direct evidence in the leaked emails is lacking, it is implicit throughout that the purpose of running pro-Obamacare ads was to help pass Obamacare. The specific request for \emph{``Harry and Louise ads''} is also suggestive, given their particular connection to healthcare reform. 

Finally, the emails show that this was in exchange for controls on pharmaceutical drug prices not becoming law. An email from Democrat strategist Steve McMahon to \emph{PhRMA} lobbyist Bryant Hall (dated 3 June 2009, \cite{memo_8_June}) acknowledged, \emph{``I told them there were bright lines, but PhRMA was serious about reform. We agreed that we would talk on a regular basis about what they wanted and what we might be able to do. I also told them if the lines were violated the money would go elsewhere and `that wouldn't be good'.''} This understanding of a quid pro quo is critical to our model. 

Note that our motivating example looks at a case where the government uses resources from a deal \emph{directly} to fight other interest groups. It is easiest to see the mechanism in action in this direct case. But the model also allows \emph{indirect} use of the resources -- the government can use them to free up political capital from something unrelated, and use that other political capital in contests. 

We also generate precise predictions of which groups a government cedes contests to. If the special interest groups' preferred policies are not too harmful, then the government concedes all contests. It prefers to acquire as much political capital as possible by doing deals, as we assume that having political capital brings some benefit to the government. On the other hand, if special interest groups' preferred policies are sufficiently harmful, then the government cedes exactly the number of contests needed to allow it to win all remaining contests with certainty. And it cedes to interest groups whose preferred policy outcome is least harmful to society. In stark contrast, if the government cannot do deals, it adopts a `scattergun' approach and spreads its resources across all possible contests, in proportion to the importance of each contest. 

The model explains why even governments that appear weak may not lose many legislative fights. It is because a weak government chooses not to engage in many fights. This highlights how the success a government has depends on the ambition of its agenda -- i.e. how many fights it picks. Further, it shows that the ability to make deals dramatically reduces the amount of lobbying special interest groups do directly. However, special interest groups make contributions to the government -- which increases the amount of political capital the government spends on contests. Unsurprisingly, a government that starts with less political capital takes more contributions from special interest groups. But it also spends less on contests. \\

\noindent \textbf{Related Literature.} 
This paper relates to two strands of literature; lobbying and conflict. Our key contribution is to allow the government to make deals with special interest groups; ceding on some policies in exchange for resources that can then be used to help secure other policies the government wants. 
This adds to the lobbying literature by providing a stylised way of capturing the back-room deals that governments might do with special interest groups. 

Most of the existing literature on lobbying uses either information or agency issues as the core friction (see \cite{bombardini2020empirical} for a recent review). In the former case, largely following \cite{crawford1982strategic}, special interest groups have better information than the government and affect what the government wants to do through strategic communication. See for example, \cite{esteban2000wealth, bennedsen2002lobbying} and \cite{roberti2014lobbying}. 

The latter, largely following \cite{grossman1994protection}, models agency issues as the government caring both about social welfare and about direct payoffs to itself. Special interest groups influence what the government wants to do by offering money. Included in this second strand is work that models interactions between different interest groups as a conflict (see for example \cite{esteban1999conflict} and \cite{kang2016policy}). As in our paper, contest success functions are used to model the conflict, but the government is not a party in the contest.

A notable exception is \cite{franke2015conflict}. While theirs is a more general study of conflict on a network (including our set-up where one central agent engages in contests with many peripheral agents), it can cover an application to lobbying. However, they have no notion of conceding contests. In the conflict literature, the ability to cede contests in exchange for a transfer of resources represents a novel addition to Colonel Blotto games (an important class of games in the conflict literature).\footnote{Colonel Blotto games were first introduced to economics by \cite{borel1921theorie, borel1953theory} and \cite{blackett1954some, blackett1958pure}. They typically involve two players spending resources on a number of bilateral contests \citep{roberson2006colonel, hart2008discrete, kovenock2012battlefields}, although n-player variants have been considered more recently \citep{boix2021multiplayer}. There is also a large literature that considers negotiations prior to contests, but not using the Colonel Blotto framework. See for example \cite{sanchez2009conflict, sanchez2012bargaining}, \cite{herbst2017balance} and \cite{ghosh2019intimidation}.} 
We find this has a dramatic impact on equilibrium outcomes. The government moves from spreading its resources across all contests to fighting only where it will win for certain. This suggests the ability to cede contests has potentially important implications for battles on multiple fronts beyond this application to lobbying.

\section{Model}\label{sec:model}
There is one risk-neutral \emph{government}, $G$, endowed with $K^G > 0$ units of \emph{political capital}. There is a measure one of \emph{laws}, $i \in I = [0,1]$, and one \emph{interest group} per law, also indexed $i$. The game has two stages. 

\paragraph{Stage 1: Bargaining.} The government simultaneously bargains with each interest group $i$ over the respective law. If a \emph{deal} is reached, then interest group $i$ makes a \emph{contribution} $B_i > 0$ (which is a transfer of resources to the government), and the government concedes law $i$, losing it with certainty.

We assume a deal is done if and only if there exists a deal (and hence some $B_i > 0$) that both agents prefer over no deal. The size of the contribution is determined by a bargaining function $\Gamma_i: \mathbb{R}^2 \to \mathbb{R}$, where $0$ represents `no deal'. Specifically, $\Gamma_i(b_i^{min}, b_i^{max}) = f(b_i^{min}, b_i^{max}) \cdot \mathbf{1}\{ b_i^{min} \leq b_i^{max} \}$, with $f(b_i^{min}, b_i^{max}) \in [b_i^{min}, b_i^{max}]$, and where $b_i^{min} \geq 0$ and $b_i^{max} \geq 0$ are the (endogenously determined) contributions that would leave the government and interest group, respectively, indifferent between a deal and no deal. Denote the set of interest groups that do a deal $\mathcal{B} = \{i : B_i \neq 0\}$.\footnote{It is also possible to replace this cooperative bargaining with a non-cooperative framework, but doing so adds more notation without affecting the results.}
These contributions add to the government's stock of political capital, so it has total political capital $K^T = K^G + \int_i B_i \d i$.

\paragraph{Stage 2: Contests.} First, the government chooses an amount of political capital spending, $K_i \geq 0$, for each contest $i$, subject to $\int_i K_i \d i \leq K^T$. Second, each interest group $i$ simultaneously chooses a quantity of lobbying $L_i \geq 0$. The government moves before the interest groups. So the government's strategy in the contest stage is $K: \mathbb{R}^I \to \mathbb{R}^I$, and each interest group's strategy is $L_i: \mathbb{R}^I \times \mathbb{R}^I \to \mathbb{R}$.

\paragraph{Outcome of laws.} The outcome of each law is binary: either the government wins (``\emph{good}''), or it loses (``\emph{bad}''). If the government does a deal over law $i$, it concedes law $i$, and so wins with probability $0$. Otherwise, the probability the government wins is determined by a standard Tullock contest function \citep{tullock1980, garfinkel2012oxford}.\footnote{\cite{skaperdas2012persuasion} provide micro-foundations showing it is an outcome of a game where contestants present evidence to persuade an audience. The Online Appendix provides similar micro-foundations directly in our setting.} 
That is:
\begin{align}\label{eqn:contest fn}
    p_i =
    \begin{cases}
    0 \ &\text{ if } i \in \mathcal{B} \\
    \frac{K_i}{K_i + L_i} \ &\text{ if } i \notin \mathcal{B} \text{ and } K_i + L_i > 0 \\
    0.5  \ &\text{ otherwise. }
    \end{cases}
\end{align}

\paragraph{Preferences.} The value of winning a law $i$ to the interest group is $\pi_i \geq 0$, and the value of winning to the government is $\alpha \pi_i$. We assume $\pi_i$ is finite for all $i$ and $\alpha > 1$. For clarity, we say $\pi_i$ is the \emph{importance}, and $\alpha$ is the \emph{harmfulness}, of a law. Interest group $i$ cares only about its own expected payoff from the outcome of law $i$, less the cost of lobbying it does and contribution it makes (if any). Therefore it has preferences:
\begin{align}\label{eq:firm prefs}
    u_i = \underbrace{(1 - p_i) \pi_i}_{\text{gains from winning}} - \quad \underbrace{(L_i + B_i)}_{\text{costs}}.
\end{align}

The government cares about three things. First, it pays a cost $\alpha \pi_i$ when it loses law $i$. Second, any unused political capital has a value: if it is not spent on contests, it can be spent on some other government project (outside the model). Third, it cares about the lobbying done by interest groups. Putting these together yields preferences:
\begin{align}\label{eq:govt prefs}
    u_G = \underbrace{- \int_i \alpha \pi_i (1 - p_i) \d i}_{\text{losses from losing contests}} 
    + \underbrace{\left( K^T - \int_i K_i \d i \right)}_{\text{unspent political capital}}
    - \underbrace{ \int_i L_i \d i.}_{\text{lost surplus from lobbying}}
\end{align}

There are two reasons for caring about lobbying by interest groups. One, lobbying consumes real resources that could have been spent on something else, including the time of those the interest group is seeking to influence.\footnote{From the viewpoint of social welfare, all resources spent contesting laws are wasted \citep{krueger1974political, tullock1980, aidt2016rent, congleton2018oxford}. To see this here, notice that the contest function is homogeneous of degree zero, so any $p_i$ can be achieved by spending an arbitrarily small quantity of resources.} 
This takes the view that lobbying is a socially damaging activity, even when interest groups are not successful in getting their way.
Two, it is important for tractability.\footnote{From a technical standpoint, it helps pin down the contributions paid in the first stage by making the minimum contribution the government is willing to accept linear in the quantity of political capital it has, and exactly equal to $\pi_i$ in the case where the government has enough political capital to win all contests for sure. Absent this feature, it would be difficult to pin down the outcome of the bargaining stage without more structure on the bargaining process.}

\paragraph{Solving the game.} We will look for Pure Strategy Subgame Perfect Nash Equilibria of this game. For convenience, we call them `equilibria' throughout. We will solve by backwards induction.

\subsection{Discussion} 
Before going further, it is helpful to discuss what contests, and the resources devoted to them, look like in our setting. A contest between a special interest group and the government is a battle for the votes of the members of a country's legislature. Outcomes of laws are binary. So think of this as legislators voting for or against a bill. An agent wins the contest when a majority of legislators vote their way.
To avoid confusion, we call these outcomes `good' and `bad'. The `good' outcome is in the general interest, and `bad' benefits a special interest group. For example, a law could either impose protectionist tariffs on imported steel (`bad') or not (`good').

Political capital [resp. lobbying] is anything that can be used to increase the probability legislators vote for the government [resp. interest group]. To fix ideas, we could think of political capital as the government's ability to provide ministerial jobs/cabinet positions, seats on committees, or visits to the legislator's district by senior officials. Similarly, lobbying could take the form of well-paid second jobs for legislators, relocation of factories to their districts, or campaigns to build support among local voters. Note that in our model, the government and legislators are distinct agents.

The model assumes that if the government does a deal, then the special interest group wins without spending anything. This is particularly intuitive if the interest group is defending the status quo: non-participation by the government involves it not suggesting a change to the law at all. The interest group does not need do any lobbying to convince the legislature here, because there is nothing for the legislature to vote on.

The contributions that special interest groups make when they do a deal with the government take the same form as lobbying. The difference is that the government chooses how and where to spend these resources. The government \emph{controls} these resources, and this is what matters.

\section{Results}\label{sec:results}
\subsection{Second stage}
In the second stage of the game, the contributions, $B_i$, are fixed. Therefore, the set of laws available for contest, $\mathcal{B}^c$, and the government's total political capital, $K^T$, are also fixed. For a law where a deal has already been made, the government loses for certain, and so neither agent spends resources on the contest in question.
For a law that is available for contest, the special interest group and government spend up to the point where the marginal expected benefit of doing so equals the marginal cost. 

Because the government moves before the special interest group, it is able to anticipate how its choice of political capital, $K_i$, affects lobbying, $L_i$.\footnote{This means that the effect $K_i$ has on $L_i$ is considered as part of the benefit of spending political capital. Due to the shape of the contest function, the special interest group $i$'s best response function is non-monotonic in $K_i$. This is a standard implication of the Tullock contest function (see, for example, \citet[Ch. 2]{konrad2009strategy}).} 
As the government's marginal cost of spending political capital is the same for all contests, it must therefore equate marginal benefits across contests. 
Note that the only benefit of spending political capital is to win contests. So once the government spends enough to guarantee victory, there is no reason to spend beyond that point (as spending is always costly).

Together, these observations pin down behaviour in the second stage. One special case where we might be interested in the second stage by itself is a world where deals are not permitted at all. There, the government adopts a `scattergun' approach, spreading itself across all contests in proportion to the costs of bad laws. 

\begin{prop}\label{prop:eq wo bribes}
Fix $B_i = 0$ for all $i$. If $K^T < \int_i \pi_i \d i$, there is a unique equilibrium:
\begin{align*}
K_i^* = \pi_i Z \quad , \quad L_i^* = \pi_i ( Z^{0.5} - Z)  \quad \text{ for all } i, \text{ where } Z \equiv \frac{K^T}{\int_i \pi_i \d i}. 
\end{align*}
Otherwise, there is a unique equilibrium: $K_i^* = \pi_i$ , $L_i^* = 0$ for all $i$.\footnote{Notice that because we have fixed $B_i = 0$ for all $i$ here, $K^T = K^G$ in this special case. But in general, $K^T \geq K^G$ because the government's total political capital is its initial endowment, plus what it receives in contributions.}
\end{prop}

Due to the scattergun approach, the government will only win a contest with certainty if it wins \emph{all} contests with certainty. That only happens when the government has high enough quantities of political capital and is able to completely deter all lobbying.

\subsection{First stage}\label{sec:endogenous bribes}
In the first stage, the government and interest groups can do deals. What deals parties find acceptable depends on what would happen if a deal were not reached and they had to contest the law in the second stage. In other words, both the government and the special interest group have an outside option of not doing a deal and fighting. They will only do a deal that is better than their outside option.

When the government does a deal with interest group $i$, it takes a contribution which increases its stock of political capital. It also shrinks the set of laws that are available for contest, which both frees up political capital and allows it to focus its spending on a smaller set of contests. All three of these forces improve the government's position in any other contests that happen. This means its outside option improves when negotiating deals with other interest groups -- so the minimum contribution it finds acceptable rises. But this improvement in the government's position also worsens the other interest groups' outside options. They know that if they fail to reach a deal, they will do worse in the subsequent contest. This increases the maximum contribution they find acceptable.

It turns out that whenever the government does not yet have enough political capital to win all remaining contests with certainty, the minimum contribution that is acceptable to the government is less than the maximum contribution that is acceptable to the interest group. So there is always room for a deal to be done. This means deals happen. In turn, this improves the government's outside option, and worsens interest groups' outside options, in other negotiations.

So the government will do deals at least to the point where it has enough political capital (including contributions) to win all remaining contests with certainty. Past this point, additional deals still accrue additional political capital -- both through contributions and saving on spending in the contests. But there is no longer any benefit to spending it in other contests, as there are all already won with certainty.

Whether the government wants to continue to do deals past this point depends critically on the harmfulness, $\alpha$, of bad laws. The benefit of doing a deal past this point is $B_i + \pi_i$ (the contribution it receives plus the political capital it frees up),\footnote{Note that in the region where the government has enough political capital to win all (remaining) contests with certainty, it spends $K_i = \pi_i$ on contest $i$. Hence doing a deal frees up $\pi_i$ units of political capital.} and the cost is $\alpha \pi_i$. So the government is willing to do a deal when $B_i \geq (\alpha - 1) \pi_i$. And the interest group is willing to do a deal whenever $B_i \leq \pi_i$, because it will lose the contest for sure if a deal is not reached.

Clearly when $\alpha \leq 2$, there is always a mutually agreeable deal to be done. So the government does a deal with every interest group. Notice that without pinning down the negotiation process, we cannot say how large the contributions will be in this case. All we can say is that they will be somewhere between the minimum amount the government is willing to accept and the maximum amount the interest group is willing to offer.\footnote{One way to pin them down would be to assume the government makes take-it-or-leave-it offers to special interest groups (who themselves can only accept or reject the offer). Then the bribes are pinned down to exactly $\pi_i$ -- the maximum the interest groups are willing to pay.}

In contrast, when $\alpha > 2$ there are no further mutually acceptable deals once the government has enough political capital (including contributions) to win all remaining contests with certainty. So the government does deals up to this point and no further. In this case, the size of the contribution is pinned down. Special interest groups pay a contribution equal to the full importance, $\pi_i$, of their law. They are willing to pay this because failing to reach a deal results in them losing the subsequent contest for certain. And this is the minimum amount the government is willing to accept.

\begin{prop}\label{prop:eq with bribes}
\noindent If $\alpha \in (1,2]$, in all equilibria: $B_i \in [(\alpha - 1) \pi_i , \pi_i]$, $K_i = 0$, $L_i = 0$ for all $i \in I$. 

\noindent If $\alpha > 2$ and $K^G \geq \int_i \pi_i \d i$, there is a unique equilibrium: $B_i = 0$, $K_i = \pi_i$, $L_i = 0$ for all $i \in I$.

\noindent If $\alpha > 2$ and $K^G < \int_i \pi_i \d i$, then in all equilibria: 

    (1) $B_i = \pi_i$, $K_i = 0$, $L_i = 0$ for all $i \in \B$ 
    
    (2) $B_i = 0$, $K_i = \pi_i$, $L_i = 0$ for all $i \in \B^c$
    
    (3) $\B$ is such that $\int_{j \in \B^c} \pi_j \d j = K^G + \int_{j \in \B} \pi_j \d j$.
\end{prop}

This result tells us several things. First, the government will do enough deals to ensure that it wins all contests it does engage in with certainty. Second, when bad laws are not sufficiently harmful (i.e. $\alpha \leq 2$) the government does deals with every single interest group. It prefers to build up a large `war chest' of political capital, but to use that political capital on something outside of the model. Third, when bad laws are sufficiently harmful (i.e. $\alpha > 2$) the government stops doing deals as soon as it has sufficient political capital to win all remaining laws with certainty. This also forces special interest groups to make large contributions.

When $\alpha < 2$, we know that the government does deals with every special interest group, but we do not know the size of the contributions. Nevertheless, we do know the minimum equilibrium contributions are increasing in the harmfulness of bad laws -- the price of protection must be high enough to compensate the government for the costs of bad laws.

In contrast, when $\alpha > 2$ we know how much protection is for sale and at exactly what price. But we cannot say who buys it. The result pins down the mass of interest groups (weighted by the importance, $\pi_i$, of their laws) that do deals, but not their identities. This is because the government does not care about the identity of interest groups \emph{per se} -- only about the costs and benefits of cutting a deal. The fact that the harmfulness of laws, $\alpha$, is constant means that the government is indifferent between doing deals with any interest group (obviously conditional on the contribution being $\pi_i$). The importance of the law (i.e. the size of $\pi_i$) does not matter. Larger $\pi_i$ increases the cost of doing a deal. But it also increases the benefit -- as the interest group is willing to make a larger contribution. These two forces perfectly offset one another, leaving the government indifferent.\footnote{That these two forces perfectly offset -- resulting in the `importance' of individual laws playing no role in the government's choice of which deals to make -- relies on our assumption that the government cares about the lobbying done by interest groups. This is in part because caring about lobbying by interest groups in this way gives the government an extra reason to do deals and avoid contests.}

An obvious, and yet important, implication of this result is that the government accepts weakly more bribes, and the total resources spend on contests is weakly lower, when it has a smaller initial endowment of political capital.

\begin{cor}
When the government has a smaller initial endowment of political capital, $K^G$, it (i) accepts more bribes and (ii) spends less on contests.
\end{cor}

The total contributions accepted are decreasing in $K^G$ because a larger initial endowment of political capital means the government does not need to do as many deals before it has enough to win all remaining contests with certainty. Fewer deals then means the government engages in more contests, and  spends more resources on them overall. And although the government fights more, it is better off when it has higher $K^G$. This is because it is better for the government to fight and win for certain than to concede and take a contribution of $\pi_i$.

Note that while \Cref{prop:eq with bribes} involves no lobbying in equilibrium, this prediction is not robust. It relies on the simplifying assumption that when there is a deal the interest group wins for certain without any lobbying. Assuming instead that an interest group needs to do some lobbying to get its preferred law passed, even if the government does not participate, would generate lobbying in equilibrium. Intuitively, this might be to ensure swift passage of the legislation (this is discussed briefly in the Online Appendix). 
But even without this extension, the model still predicts interest group money in politics -- coming from the contributions they make.

\subsection{Extension: Heterogeneous Harms}\label{sec:heterogeneous harms}
Relaxing the assumption that all laws are equally harmful pins down the identities of the special interest groups that do deals. The outcome is that the government does deals with interest groups whose laws are less harmful. It still wants to take contributions up to the point where it will completely deter lobbying in all remaining contests, and so will still extract the full value of $\pi_i$ as a contribution. However, it now plays a cut-off strategy. It will make a deal with \emph{every} interest group with an $\alpha_i$ below a threshold, and from \emph{no} special interest groups above the threshold.

To show this, we need to put some formal restrictions on the $\alpha_i$. Assume that $\alpha_i \in (2,\infty)$ for all $i$ and that they can be represented by some atomless Cumulative Distribution Function. Note that we are assuming $\alpha_i > 2$, rather than $\alpha > 1$ as in the main model. This is in the interest of tractability. In Proposition 1, we were not able to pin down the exact contributions when $\alpha \in (1,2]$. It is helpful to avoid that issue here. With this, we can now pin down who makes contributions.

\begin{prop}\label{prop:eq with het harms}
Suppose $\alpha_i > 2$ for all $i$, and are distributed according to some atomless CDF. There is a unique equilibrium: for some $\overline{\alpha}$

(1) $B_i^* = \pi_i$, $K_i^* = 0$ \ and $L_i^* = 0$ if $\alpha_i < \overline{\alpha}$, and

(2) $B_i^* = 0$, \ $K_i^* = \pi_i$ and $L_i^* = 0$ if $\alpha_i \geq \overline{\alpha}$.
\end{prop}

This tells us exactly who does deals and makes contributions -- the government does deals and cedes contests over the laws that it finds least objectionable. This is because interest groups only care about their own benefits, $\pi_i$, and so are all equally willing to make a contribution (conditional on $\pi_i$ and $K^T$). The government effectively faces constant benefits but heterogeneous costs from taking contributions. It is natural to seek to minimise costs. The other features of this result are the same as in \Cref{prop:eq with bribes}. Note, the fact that the government plays a cut-off strategy is a result, not an assumption.

\subsection{Key Predictions}
When the government can do deals, the model generates three key stylised predictions. First, the government does not lose contests. It either cedes the contest before it even happens, or wins. This is consistent with anecdotes that it is rare for governments to lose votes in their parliament or congress \citep{mezey1979comparative, mayhew2004congress, russell2016policy}.

Second, the amount that special interest groups spend on politics will be substantially smaller than the total value of outcomes if the government is fairly strong (in terms of its initial endowment of political capital). Note that spending on politics can either come through lobbying, or through contributions to the government. Both involve the same kind of resources; what distinguishes them is who directs their spending and what they are used to achieve. 

This provides a possible explanation for `Tullock's Puzzle' -- an observation, often first attributed to \cite{tullock1972purchase}, that special interest groups spend little on lobbying relative to the very large consequences of government decisions. A standard explanation for this is the collective action problem \citep{olson1965} -- with many firms benefiting from a given law, they may struggle to coordinate and avoid free-riding. In contrast, our model abstracts away from possible coordination issues -- there is only one interest group that cares about each law, and it chooses lobbying directly. Low spending on lobbying in our model is instead driven by the way the government does deals and picks its battles.

However, the model can also rationalise increasing spending by special interest groups. Spending by special interest groups rises if the government becomes weaker relative to the total value of contestable outcomes. This could be because the government has less political capital initially (lower $K^G$), or because more issues become contestable or those issues become more important (higher $\int_i \pi_i \d i$).

Third, the government takes contributions from those whose laws are least harmful (i.e. have the lowest $\alpha_i$). By doing so it can win contests over the more harmful laws. %That is, it does deals with interest groups whose laws have a low value of $\alpha_i$ and successfully contests laws with a high value of $\alpha_i$. %% MAYBE REINSTATE
This prediction has parallels to \cite{becker1983theory, becker1985public} in that competition between interest groups leads to a constrained social optimum. While outcomes are far from ideal, the government would need a larger endowment of political capital to improve welfare. 
Related to this, the model predicts that the government will engage in \emph{more} contests when it starts with more political capital -- and the government will win them all.

Simple comparative statics lead to two concrete predictions about what happens when a government becomes more powerful (in the sense that it has a greater endowment of political capital). First, it picks more battles. It then does correspondingly fewer deals with special interest groups -- a stronger government offers special treatment to fewer interest groups. Second, the `price' of that special treatment does not change. Because the government manages its affairs through the quantity of deals it does, it is always in a strong enough negotiating position against any individual special interest group to be able to extract large contributions. 

These stark results are driven in large part by simplifying assumptions we make to get clean results. In reality, we would not expect that the government never loses, or that it extracts the full value $\pi_i$ from special interest groups. Rather, the model seeks to highlight the importance of picking your battles and the benefits of back room deals.

\subsection{Extensions}\label{sec:extensions discussion}
We have kept the model very simple in order to focus on agents' behaviour and to highlight the role of deals. There are of course a number of natural variations to our setup. First, it might be that in the contests, agents choose political capital/lobbying simultaneously. In this case, government will take fewer contributions, and will not completely deter lobbying.

Second, dropping the outside option of political capital has little effect (although it removes the case where the government does deals in order to build up its war chest of political capital beyond the point needed to win all remaining contests for sure). Intuitively, this might correspond to a setting where officials' time and committee seats can only be used to persuade members of the parliament to vote on laws, and have no other potential use.
Finally, relaxing the benchmark assumption that legislation is perfectly effective has no meaningful impact. Suppose instead that legislation attenuated the costs and benefits by some law-specific multiple. Then the government considers only the portion of the costs it can affect, and behaviour is otherwise unchanged.
\section{Conclusion}\label{sec:conclusion}
This paper puts conflicts and power at the centre of the lobbying process to shed light on how a government spends political capital and sells special treatment. It considers a world where the government makes many different policy decisions and has limited resources. We find that when a government cannot do deals with special interest groups, it adopts a scattergun approach and spreads its political capital across all contests. In stark contrast, a government that can do deals will pick its battles carefully. It will only fight where it can win and will do a deal in all other cases. 

This mechanism can explain why even weak governments do not appear to lose legislative battles very often. They find it better to not fight at all and gain what they can from cutting a deal. Weak governments fight less, rather than less successfully. It also highlights the important role that backroom deals can play in the wider political landscape. Bilateral deals have indirect impacts on other interest groups, by putting the government in a better position to win other contests. 

Our model also highlights a tension between efficiency and equal treatment. This is because it shows how the special treatment given to certain interest groups can be part of a bigger picture that improves overall outcomes for the government and society, even if individual deals in isolation harm society. So rules or social norms that seek to prevent the government giving special treatment to some groups may not always be optimal. More nuanced rules, allowing special treatment but preventing members of the government from using contributions to enrich themselves, might be better.

The welfare implications of having the government do deals depends on the harmfulness of bad laws. When they are sufficiently harmful, the government will only do deals in order to help win other contests. In this case, allowing deals unambiguously improves welfare. However, when bad laws are less harmful, the government may do deals to build up its stock of political capital beyond the level needed to win contests. This excessive build-up of political capital can reduce welfare. 

The key contribution is to show how allowing the government to do deals with special interest groups has an important impact on behaviour and outcomes. It shows that, rather than behaving in a scattergun manner, the government picks its battles carefully. It only fights where it can win.

\newpage
\singlespacing
\bibliographystyle{abbrvnat}
\addcontentsline{toc}{section}{References}
\bibliography{bib}
\appendix

%\doublespacing
%\newpage
\section{Proofs of Main Results}\label{sec:proofs}
\subsection*{Proposition 1}
It is more convenient to state and then prove a generalised version of \Cref{prop:eq wo bribes}, as we will need it for later results. 

\begin{lem}
Fix a set of contributions $B_i$ (and hence $\B$ and $K^T$). Then there is a unique equilibrium:

(A) \ If $B_i \neq 0$: $L_i^* = 0$, and $K_i^* = 0$.

(B) \ If $B_i = 0$: $L_i^* = (\pi_i K_i^*)^{0.5} - K_i^*$ if $K_i^* \leq \pi_i$ and $L_i^* = 0$ otherwise, and $K_i^* = \pi_i Z$ if $Z \leq 1$ and $K_i^* = \pi_i$ otherwise,

where $Z = \frac{K^T}{\int_{j \in \B^c} \pi_j \d j}$
\end{lem}

\begin{proof}
It will be clearer to do the derivation for a finite number ($n$) of special interest groups, and then take the limit as $n \to \infty$. 
\textbf{Case A.} Trivial. $p_i = 0$ for all $K_i$, $L_i$. So $\frac{d u_i}{d L_i} < 0$, so $L_i^* = 0$. Similarly, $\frac{d u_G}{d K_i} < 0$, so $K_i^* = 0$. 

\textbf{Case B.} Solve backwards. Moving last, interest group $i$ maximises $u_i = \left( 1 - \frac{K_i}{K_i + L_i} \right) \pi_i - L_i$, subject to $L_i \geq 0$. We solve this in the standard way: take the First Order Condition, apply the standard quadratic formula, and account for the fact that $L_i \geq 0$.\footnote{The First Order Condition is: $K_i \pi_i (K_i + L_i)^{-2} - 1 = 0$. When solving for $L_i$ using the quadratic formula, it is straightforward to spot that the smaller of the two solutions is always negative, and so can be ignored. The larger of the two solutions is not always positive, so we choose the larger of (i) the solution and (ii) zero.} 
Doing so yields: $L_i^* = (\pi_i K_i)^{0.5} - K_i$ if $K_i \leq \pi_i$, and $L_i^* = 0$ otherwise.
Next, the government maximises $u_G$ subject to the constraint that it cannot spend more than $K^T$. So set up the Lagrangian for the government:
\begin{align}
    \mathcal{L} &= \sum_{i \in \mathcal{B}^c} \left([ - \alpha \pi_i (1- p_i) - L_i - K_i] \frac{1}{n} \right) + K^T - \mu \left(\sum_{i \in \mathcal{B}^c} K_i \frac{1}{n}  - K^T \right).
\end{align}
Because the government moves before the interest groups, we substitute in the interest groups' best response functions. These best response functions have two cases.\footnote{Note that we have omitted agents $j \in \B$ from the summation. This is because we know $p_j = 0$ and (hence) $K_j^* = 0$, $L_j^* = 0$, so the summation over $j \in \B$ is summation over a constant.} \textbf{First}, if $K_i \geq \pi_i$, then $L_i^* = 0$. So $p_i = 0$ and hence
\begin{align}
    \mathcal{L} = \sum_{i \in \B^c} \left[ - \alpha \pi_i - K_i \right] \frac{1}{n} + K^T - \mu \left(\sum_{i \in B^c} K_i \frac{1}{n} - K^T \right) 
    \quad \implies \quad 
    \frac{d \mathcal{L}}{d K_i} = \frac{-(1 + \mu)}{n} < 0,
\end{align}
so spending more than $\pi_i$ in a contest cannot be optimal. \textbf{Second}, if $K_i < \pi_i$, then $L_i^* = (\pi_i K_i)^{0.5} - K_i$. So $p_i = (K_i / \pi_i)^{0.5}$ and hence
\begin{align}
    \mathcal{L} &= \sum_{i \in \B^c} \left[ - \alpha \pi_i + (\alpha - 1) (K_i \pi_i)^{0.5} \right] \frac{1}{n} + K^T - \mu \left(\sum_{i \in B^c} K_i \frac{1}{n} - K^T \right), \\
    \frac{d \mathcal{L}}{d K_i} &= \frac{1}{2} \frac{1}{n} (\alpha - 1) \left( \frac{\pi_i}{K_i} \right)^{0.5} - \mu \frac{1}{n}. \label{eq:MBK}
\end{align}

Now suppose $\mu = 0$ (i.e. the government's budget constraint does not bind). Then $K_i < \pi_i$ implies that $\frac{d \mathcal{L}}{d K_i} > 0$. So if the budget constraint does not bind, then we must have $K_i \geq \pi_i$. But we know from above that $K_i > \pi_i$ cannot be optimal. Therefore, if the budget constraint does not bind, we must have $K_i = \pi_i$ for all $i$. This means the budget constraint binds if and only if $\int_{i \in \B^c} \pi_i \d i < K^T$.

Then, when the budget constraint \emph{does} bind, we must have $K_i < \pi_i$, and hence $\frac{d \mathcal{L}}{d K_i} = 0$, for some $i \in \B^c$. But when the budget constraint binds, the marginal benefit of spending must be equated across all contests. That is:
\begin{align*}
    \frac{d \mathcal{L}}{d K_i} = \frac{d \mathcal{L}}{d K_j} \implies \frac{K_i^*}{K_j^*} = \frac{\pi_i}{\pi_j} \ \text{ for all } i,j \in \B^c. \label{eq:K_star_normal}
\end{align*}
Since the budget constraint binds, we must have $\sum_{j \in \B^c} K_i^* \frac{1}{n} = K^T$. Summing over $j \in \B^c$ and rearranging yields:
\begin{align}
    %\sum_{j \in \B^c} K_i^* \frac{1}{n} = \sum_{j \in \B^c} \frac{\pi_i}{\pi_j} K_j^* \frac{1}{n}
    %\implies 
    K_i^* = \frac{\pi_i K^T}{ \frac{1}{n} \sum_{j \in \B^c} \pi_j }.
\end{align}
Finally, notice that as $n \to \infty$, $\frac{1}{n} \sum_{j \in \B^c} \pi_j \to \int_{j \in \B^c} \pi_j \d j$.
As required, this yields $K_i^* = \pi_i Z$, where $Z = \frac{K^T}{\int_{j \in \B^c} \pi_j \d j}$.
\end{proof}
\vspace{5mm}

Note that $K^T$, and hence $Z$, depend on the deals that get done (and the contributions paid). In Lemma 1, the contributions are fixed, so $K^T$ and $Z$ are exogenous. But when determining the \emph{equilibrium} contributions, $K^T$ and $Z$ are equilibrium objects. In order to understand the equilibrium contributions, it is first helpful to look at interest groups' willingness to pay contributions and the government's willingness to accept them.

%\subsection*{Lemma 2}
\begin{lem}
Let $b_i^{max}$ denote the largest contribution interest group $i$ is willing to pay, and let $b_i^{min}$ denote the smallest contribution the government is willing to accept, as part of a deal.

(A) \ $b_i^{max} = \pi_i (2 Z^{0.5} - Z)$ if $Z \leq 1$, and $b_i^{max} = \pi_i$ otherwise.

(B) \ $b_i^{min} = \pi_i Z$ if $Z \leq 1$, and $b_i^{min} = \pi_i (\alpha - 1) $ otherwise.

(C) \ $b_i^{max} \geq b_i^{min}$ if and only if either $Z \leq 1$ OR $\alpha \leq 2$
\end{lem}

\begin{proof}
Suppose there is an offer for a deal in exchange for a contribution of $B_i$. 

\textbf{(A)} If the special interest group accepts, it gains (relative to rejection) $L_i^* + p_i^* \pi_i$ (saving on lobbying costs, plus an increased probability of getting a bad law), and it loses $B_i$. So it accepts if and only if $B_i \leq p_i^* \pi_i + L_i^*$. That is, $b_i^{max} = p_i^* \pi_i + L_i^*$.
If $Z \leq 1$, then (from Lemma 1) $K_i^* = \pi_i Z$, and $L_i^* = (\pi_i K_i^*)^{0.5} - K_i^*$. Substituting these into the equation for $b_i^{max}$ yields the result.\footnote{To see this, notice that when $Z \leq 1$: $p_i^* \pi_i + L_i^* = \frac{\pi_i K_i^*}{(\pi_i K_i^*)^{0.5} - K_i^* + K_i^*} + (\pi_i K_i^*)^{0.5} - K_i^* = 2(\pi_i K_i^*)^{0.5} - K_i^*$, then substitute in $K_i^* = \pi_i Z$ and simplify.} 
If $Z > 1$, then (from Lemma 1) $K_i^* = \pi_i$, and $L_i^* = 0$. The result is then clear.

\textbf{(B)} If the government accepts, it receives $- \alpha \pi_i$ from the contest, and has $K^T + B_i$ to spend on other contests. If it rejects, it receives $-\alpha \pi_i + p_i^* \alpha \pi_i - K_i^* - L_i^*$ from contest $i$, and has $K^T - K_i^*$ to spend on other contests. So accepting brings a net benefit of $- p_i^* \alpha \pi_i + K_i^* + L_i^*$ from contest $i$ and a benefit of $K_i^* \mu + B_i \hat{\mu}$ from being able to spend more in other contests, where $\mu$ is the Lagrange multiplier when the government fights contest $i$, and $\hat{\mu}$ when it does not. With a measure of agents, we have $\hat{\mu} = \mu$.\footnote{We prove this in the Online Appendix for completeness.}
Therefore, the government accepts if and only if $- p_i^* \alpha \pi_i + K_i^* + L_i^* + \mu B_i + \mu K_i^* \geq 0$. That is, $b_i^{min} = \frac{1}{\mu} (p_i \alpha \pi_i - L_i^* - (1+\mu)K_i^*)$.

If $Z \leq 1$, then substituting in from Lemma 1 yields $b_i^{min} = \frac{1}{\mu} ((\alpha - 1) (\pi_i K_i^*)^{0.5} - \mu K_i^*)$. It is straightforward to see from the Lagrangian that when $Z \leq 1$, we have $\mu = \frac{1}{2} (\alpha - 1) \pi_{i}^{0.5} (K_i^*)^{-0.5}$. Applying this yields $b_i^{min} = \frac{\mu}{\mu} K_i^* = K_i^*$.\footnote{To see this, notice that $(\alpha - 1) (\pi_i K_i^*)^{0.5} = (\alpha - 1) \left( \frac{\pi_i}{K_i^*} \right)^{0.5} K_i^* = 2 \mu K_i^*$. Hence we have $b_i^{min} = \frac{1}{\mu} (2 \mu K_i^* - \mu K_i^*)$.} 
Substituting in $K_i^* = \pi_i Z$ (again, from Lemma 1) completes the result.

If $Z > 1$, then again (from Lemma 1) $K_i^* = \pi_i$ and $L_i^* = 0$, and hence $p_i = 1$. Additionally, $\mu = 0$ because the government's budget constraint does not bind. The result is then clear.

\textbf{(C)} First, suppose $Z \leq 1$. From (A) we have $b_i^{max} = \pi_i(2 Z^{0.5} - Z) \geq \pi_i Z$. And from (B) we have $b_i^{min} = \pi_i Z$. Therefore $b_i^{max} \geq b_i^{min}$.
Now suppose $Z > 1$. From (A) we have $b_i^{max} = \pi_i$. And from (B) we have $b_i^{min} = \pi_i (\alpha - 1)$. Then $b_i^{max} \geq b_i^{min}$ if and only if $\alpha \leq 2$.  
Now consider the case where $Z > 1$. Then $b_i^{max} - b_i^{min} = \pi_i - (\alpha - 1) \pi_i$. Which is clearly non-negative if and only if $\alpha \leq 2$.   
\end{proof}

\subsection*{Proposition 2}
\textbf{Step 1.} \emph{First, $Z < 1$ is not possible in equilibrium.} Suppose $Z < 1$ in equilibrium. Then (a) there must exist some $i \in \B^c$, and (b) $b_i^{max} > b_i^{min}$. Claim (a) follows from the fact that if the government has done deals with every interest group, then $\B^c$ is an empty set, so $Z = \infty$. Claim (b) follows from Lemma 2. 
But this means there must be some profitable deviation to $B_i > 0$ for both the interest group and the government. Contradiction.

\textbf{Step 2.} \emph{Second, if $\alpha > 2$ and $Z > 1$ then the government cannot do any deals in equilibrium.} Suppose $\alpha > 2$ and $Z > 1$ and the government does at least one deal in equilibrium. Then (a) there must exist some $i \in \B$ and (b) $b_i^{max} < b_i^{min}$. Claim (a) is by assumption. Claim (b) follows from Lemma 2.
But this means there must be a profitable deviation to $B_i = 0$ for both the interest group and the government. Contradiction.

\textbf{Step 3.} This leaves three possibilities. (1) $\alpha \leq 2$. (2) $\alpha > 2$ and $Z > 1$ without any deals. This is the case where $K^G \geq \int_i \pi_i \d i$. (3) $\alpha > 2$ and $Z < 1$ without any deals. These are the three cases in the result.

\textbf{(1)} if $\alpha \leq 2$ then $b_i^{max} \geq b_i^{min}$ for all $i$, regardless of $Z$. So it must be that the government does a deal with every single interest group. Then the interest group is willing to pay up to $b_i^{max} = \pi_i$, and the government is willing to accept at least $b_i^{min} = (\alpha - 1) \pi_i$. So in equilibrium, a deal must be done somewhere within those bounds. This proves the first part of the result.

\textbf{(2)} if $\alpha > 2$ and $Z > 1$ without any deals, the government cannot do any deals in equilibrium (this is step 2). Since $Z > 1$, then it follows from Lemma 2 that $B_i^* = \pi_i$ for all $i \in \B$ and from Lemma 1 that $K_i^* = \pi_i$ for all $i \in \B^c$. And with no deals, $\mathcal{B}^c = I$. This proves the second part of the result.

\textbf{(3)} if $\alpha > 2$ and $Z < 1$ without any deals. The government must do deals (by step 1), but it cannot do so many deals that $Z > 1$ after the deals (step 2). So it must be that $Z = 1$ in equilibrium. 

Then it follows from Lemma 2 that $B_i^* = \pi_i$ for all $i \in \B$. And it follows from Lemma 1 that $K_i^* = 0$ and $L_i^* = 0$ whenever $B_i^* > 0$. This proves the first claim in the third part of the result.

It follows from Lemma 1 that $K_i^* = \pi_i$ and $L_i^* = 0$ for all $i \in \B^c$ (by substituting $Z=1$ into Lemma 1). That $B_i^* = 0$ for all $i \in \B^c$ is by definition. This proves the second claim in the third part of the result. 

Recall that by definition we have $Z = K^T / \int_{j \in \mathcal{B}^c} \pi_j \d j$, and $K^T = K^G + \int_j B_j \d j$. Then recall that by definition, $B_i = 0$ for all $i \in \mathbf{B}^c$. So $\int_j B_j \d j = \int_{j \in \mathcal{B}} B_j \d j$. From immediately above, we have $B_i^* = \pi_i$ for all $i \in \B$ and $Z = 1$ in equilibrium. Substituting these in to the equation for $Z$, we have
\begin{align}
    1 = \frac{K^G + \int_{j \in \mathcal{B}} \pi_j \d j}{\int_{j \in \mathcal{B}^c} \pi_j \d j }.
\end{align}
Finally, rearrange. This proves the third claim in the third part of the result.

\subsection*{Corollary 1}
From Proposition \ref{prop:eq with bribes}, we have $\int_{j \in \B^c} \pi_j \d j = K^G + \int_{j \in \B} B_j \d j$. Also from Proposition \ref{prop:eq with bribes}, we have $\int_{j \in \B} B_i \d j = \int_{j \in \B} \pi_j \d j$. And by assumption $\int_{j \in \B^c} \pi_j \d j + \int_{j \in \B} \pi_j \d j = \Pi$ where $\Pi$ is a constant. Substituting these in and rearranging, yields: $\int_{j \in \B} B_j = \frac{1}{2} (\Pi - K^G)$. Therefore
\begin{align*}
    \frac{d \left( \int_{j \in \B} B_j \right) }{d K^G} = -0.5.
\end{align*}
Both parts of the result follow immediately from this.

\subsection*{Proposition 3}
\textbf{Step 1.} \emph{The government does deals such that it wins all remaining contests with certainty, and does no more than this.} If follows from Lemma 2 that if $K_i^* < \pi_i$ then $b_i^{max} > b_i^{min}$. So a deal must be done. When $K_i^* = \pi_i$ is exactly the point at which it wins with certainty. Also, if $K_i^* = \pi_i$ for all $i \in \B^c$, then doing an extra deal cannot be optimal, because the benefit is $2 \pi_i$ and the cost is $\alpha_i \pi_i$ (with $\alpha_i > 2$ by assumption).\footnote{An immediate implication of this is that if the government's initial endowment of political capital is large enough for it to win all contests with certainty to begin with, then it does no deals at all. In that case, $\overline{\alpha} = \infty$ trivially.}

So it must be the case that (a) $K_i^* = \pi_i$ for all $i \in \B^c$ and (b) no further deals are done. An implication of (a) is that $B_i^* = \pi_i$ for all $i \in \B$. This is because $b_i^{max} = b_i^{min} = \pi_i$ when the government has sufficient capital to win the contest for certain if a deal is not reached.
This pins down the value of deals. It is analogous to Proposition \ref{prop:eq with bribes}.

\textbf{Step 2.} \emph{The government cannot do deals with interest groups that are more harmful than (i.e. higher $\alpha_i$) ones it fights.} We show this by contradiction. Suppose the government does deals with a set of interest groups $X$, and does not do deals with a set of groups $Y$, where both sets have positive measure. Now suppose that $\alpha_i > \alpha_j$ for all $i \in X$ and all $j \in Y$.

Then the government has a profitable deviation where it stops doing deals with some members of $X$, and starts doing deals with some members of $Y$, in such a way as to keep $\int_{i \in \B} \pi_i$ constant.
This would increase the government's utility by reducing the value of the payoffs it concedes, without affecting its budget. Hence this represents a profitable deviation. This continues until one of the two sets is exhausted.\footnote{Either there exists $Y' \subseteq Y$ such that $\int_{i \in X} \pi_i \d i = \int_{i \in Y'} \pi_i \d i$ or there exists $X' \subseteq X$ such that $\int_{i \in X'} \pi_i \d i = \int_{i \in Y} \pi_i \d i$ (or both). This means the government can swap to `better' deals while maintaining a constant quantity of contributions.}
So in equilibrium we must have $\alpha_i \geq \alpha_j$ for all $i \in \B$ and all $j \in \B^c$.
Given that $\alpha_i$ is finite for all $i$ and the CDF of $\alpha_i$'s is atomless, there then must exist some threshold $\overline{\alpha}$ such that $\alpha_i \geq \overline{\alpha}$ and $\alpha_j < \overline{\alpha}$ for all $i \in \B$ and all $j \in \B^c$.\footnote{If there were an atom, then it could be that the threshold occurs at the atom, and the government does with some, but not all, interest groups at a given value of $\alpha$. This would generate multiplicity a la \Cref{prop:eq with bribes}.}

We have shown that if an equilibrium exists, then it must take the form in the result. To complete the proof, we now show that it is in fact an equilibrium. 

\textbf{Step 3.} \emph{The threshold behaviour is an equilibrium.} No interest group can profitably deviate. (i) Interest groups paying $B_i = \pi_i$ would face $K_i = \pi_i$ if they deviated. (ii) Interest groups not doing a deal face $K_i^* = \pi_i$ and so receive $u_i = 0$. Their choice of $L_i^* = 0$ is clearly a best response. Offering $B_i > \pi_i$ clearly leads to lower payoffs. (iii) The government already has sufficient political capital to win all remaining contests with certainty, so it cannot profitably take additional bribes. The government cannot profitably take fewer bribes as shown in step 1. 
Finally, the government cannot profitably change the set of interest groups it takes bribes from -- as doing so would involve winning a contest for a less harmful law in exchange for ceding a contest for a more harmful law.

%\newpage
%\vspace*{\fill}
%\begin{center}
%    \textbf{Online Appendix}
%\end{center}
%\vspace*{\fill}

\newpage
\section*{Online Appendix}
This Online Appendix provides a range of additional material. The contents of each appendix are summarised below. \\

\noindent \textbf{\Cref{sec:contest microfoundations}} provides microfoundations for the Tullock contest function (\Cref{eqn:contest fn}). It shows that the contest function is the outcome of a game where contesting parties produce evidence. An arbiter who observes the evidence with some noise picks the winner based on the evidence she observes. \\

\noindent \textbf{\Cref{sec:coordination problems}} shows that society would want to provide the government with enough political capital to win all contests for sure. However, in line with \cite{olson1965}, we would not expect society to be able to organise the collective action needed to achieve this outcome. \\

\noindent \textbf{\Cref{sec:simulations}} shows the outcome of a set of simulations that use a finite number of special interest groups. These simulations suggest that the assumption a mass of interest groups (as in \Cref{sec:model}) does not meaningfully affect outcomes. \\

\noindent \textbf{\Cref{sec:extensions}} provides technical details of the extensions discussed in \Cref{sec:extensions discussion}. First, we consider a variant where the government and special interest groups choose political capital and lobbying simultaneously. Obviously negotiation always happens before potential contests. 
Second, we explore a variant where political capital is use-it-or-lose-it and so has no shadow value.
Finally, we show that allowing for arbitrary heterogeneity in the efficacy of laws (i.e. the extent to which a good law reduces the harms to society/benefits to the special interest group) has no meaningful impact on results. \\

\noindent \textbf{\Cref{sec:extra_proofs}} provides proofs of a claim regarding robustness to using a non-cooperative bargaining protocol in the model (made in \Cref{sec:model}), and an obvious claim made in the proof to Lemma 2.

\newpage
\section{Microfoundations of the Contest Function}\label{sec:contest microfoundations}
The Tullock contest function in \Cref{eqn:contest fn} (for when $i \notin \mathcal{B}$) is a critical component of the model. Here we provide micro-foundations, showing that it is the outcome of a binary choice problem. The relationship with the Logit model is straightforward and has been noted by \cite{jia2008stochastic} and \cite{fu2012micro}. The micro-foundations here are similar to those in \cite{skaperdas2012persuasion}. \\

\noindent \textbf{Arbiter.} There is an arbiter (possibly Nature) who evaluates \emph{evidence} for and against the proposed law and then chooses her most preferred outcome. She is unbiased in that she always chooses the outcome she observes as being better supported by evidence. However, she observes evidence with some additive noise. 
Define the evidence in favour of the government's [resp. interest group's] preferred option as $V_G$ [resp. $V_I$]. The arbiter observes this evidence, plus some noise, $\epsilon_G$ [resp. $\epsilon_I$]. Her utility is then:
\begin{align}
    U_G = V_G + \epsilon_G  \quad , \quad
    U_I = V_I + \epsilon_I,
\end{align}
where $\epsilon_G, \epsilon_I \ \overset{i.i.d}{\sim} \text{ GEV (Type I) }$. The arbiter chooses the option with higher utility. \\

\noindent \textbf{Evidence production.} The government and the interest group use political capital and lobbying effort (respectively) to produce evidence, according to the following production functions:
\begin{align}\label{eq:production fn}
    V_G = \ln(K_i)  \quad , \quad 
    V_I = \ln(L_i).
\end{align}
We can think of evidence as policy papers, think-tank reports, data analysis, economic theory, or any other type of intellectual output designed to convince a neutral arbiter. We can also take the term ``evidence'' much more loosely to mean \emph{anything} that would convince the arbiter. We could reasonably consider outright bribery and/or intimidation as evidence production in this loose sense. \\

\noindent \textbf{Arriving at the contest function.} Having set up the binary choice problem, the Tullock contest function follows straightforwardly from well-known facts about the Logit model.\footnote{For more details on the Logit and other binary choice models, see \cite{train2009discrete}.} 
The arbiter chooses the government's preferred option if and only if $U_G > U_I$ (otherwise, it sides with the interest group). Therefore, the probability the government ``wins'' is equal to $Pr(U_G > U_I) = Pr(V_G + \epsilon_G > V_I + \epsilon_I) = Pr(\epsilon_I < \epsilon_G + V_G - V_I)$. This is now identical to the Logit choice probabilities:
\begin{align}\label{eq:contest fn microfounded}
    p_i &= \frac{\exp\{V_G\}}{\exp\{V_G\} + \exp\{V_I\}}
    %p_i &= \frac{ \frac{\phi}{1 - \phi} K_i }{\frac{\phi}{1 - \phi} K_i  + L_i} \\
    \ = \frac{K_i}{K_i + L_i},
\end{align}
which is the Tullock contest function.\footnote{The functional form for evidence production determines the functional form of the contest function. For example, a linear production function would obviously lead to the less common (but still fairly well-studied) exponential form of the contest function. In the Online Appendix, we modify the evidence production function to derive a generalised version of the contest function. } 
A stark feature of this functional form is that when $K_i = 0$ then $V_G = - \infty$. This corresponds to the idea that if one side provides absolutely no evidence whatsoever (and the other provides at least some), then the decision is for the arbiter is very clear. She (the arbiter) is able to identify when there is no evidence, and avoid that option.

\section{Transfers} \label{sec:coordination problems}
Corollary 1 (in \Cref{sec:results}) tells us that the government takes more contributions when it has a smaller endowment of political capital. This means it cedes more contests, which, all else equal, is bad for society (`society' being some agents who care only about the payoffs from the outcomes of the laws). A natural question is then whether society would want to transfer its own resources to the government. 

However, we immediately run into the standard public goods problem. While each individual member of society (\emph{`citizen'}) would like to see a transfer of resources to the government to augment its political capital, they do not want to make this transfer themselves. They want to free ride. %\footnote{At least in the minds of economists, people prefer to receive the benefits of government spending without paying for it themselves.} 
While society as a whole is not modelled explicitly in \Cref{sec:model}, it is natural to think it consists of many agents. So even if free-riding issues could be resolved, we would expect significant difficulties in achieving collective action \citep{olson1965}.
One mechanism for solving the collective action problem is for citizens to vote. \\

\noindent \textbf{Voting.} There are two \emph{candidates}, $X$ and $Y$. They propose non-negative transfers $\tau_X \geq 0$ and $\tau_Y \geq 0$ respectively. Assume these transfers must be anonymous -- so they must be the same for all citizens. A candidate receives a payoff of $u_G(\tau)$ if elected, and $-A$ if not elected (where $A$ is a large finite number). We assume that the elected candidate must levy the promised transfer, and can do so costlessly.

Each citizen votes non-strategically for her preferred candidate. She randomises if she is indifferent between the two candidates.
The outcome of this is simple: competition between the two candidates push both to offer the transfer that is most preferred by citizens.

\begin{prop}
There is a unique outcome:
(i) If $\alpha \in (1,2]$, both candidates propose $\tau_X = \tau_Y = 0$, \\
(ii) If $\alpha > 2$, both candidates propose $\tau_X = \tau_Y = \int_i \pi_i \d i - K^G$, \\
(iii) Citizens are indifferent between the two candidates.
\end{prop}

\begin{proof}
As all citizens are identical in our model, it suffices to consider a single person. For concreteness, a citizen's utility function is:
\begin{align}\label{eq:C1}
    %u_c = \int_i [ - (1 - p_i) \alpha \pi_i - \tau ] \d i \\
    u_c = \alpha \int_i (p_i - 1) \pi_i \d i - \tau.
\end{align}
where $\tau$ is the transfer she makes to the government. She bears costs from laws being bad, and from making a transfer.
Substituting in optimal government behaviour by the government (from \Cref{prop:eq with bribes}) yields:
\begin{align}\label{eq:C2}
    u_c = - \alpha \int_{i \in \B} \pi_i \d i - \tau.
\end{align}
This is because we know that the government cedes contests to all $i \in \B$ (guaranteeing the bad law), and wins for certain against all $i \in \B^c$. There are two cases to consider. 

\textbf{Case one:} transfers beyond the point that the government can win all contests with certainty. That is; $K^G + \tau > \int_i \pi_i$. A citizen clearly does not want to do this, because once $K^G + \tau \geq \int_i \pi_i$, then $p_i = 1$ for all $i$, and so $u_c = - \tau$.

\textbf{Case two:} transfers such that $K^G + \tau \leq \int_i \pi_i$. In this case, we know that in equilibrium: (a) $(K^G + \tau) + \int_{i \in B} B_i \d i = \int_{i \in \B^c} \pi_i$, (b) $B_i = \pi_i$ when $\alpha > 2$, and (c) $B_i \in [(\alpha - 1) \pi_i , \pi_i]$ when $\alpha \in (1,2]$. For convenience, define $\Pi = \int_i \pi_i \d i = \int_{i \in \B} \pi_i \d i + \int_{i \in \B^c} \pi_i \d i$. Some straightforward rearranging then yields:
\begin{align}\label{eq:C3}
    \int_{i \in \B} \pi_i &= \frac{1}{\chi} [\Pi - (K^G + \tau)]
\end{align}
with $\chi = 2$ when $\alpha > 2$, and $\chi \in [\alpha,2]$ when $\alpha \in (1,2]$.
Then substituting this into the citizen's utility function yields:
\begin{align}
    u_c &= \tau \left(\frac{\alpha}{\chi} - 1 \right) - \frac{\alpha}{\chi} \left( \Pi - K^G \right) \\
    \implies \frac{d u_c}{d \tau} &= \left(\frac{\alpha}{\chi} - 1 \right),
\end{align}
which is strictly positive if and only if $\alpha > 2$. 

This means that if $\alpha > 2$, the citizen's utility is increasing in the transfer $\tau$ up to the point where $K^G + \tau = \int_i \pi_i$, and decreasing beyond that. And if  $\alpha \in (1,2]$, the citizen's utility is weakly decreasing in $\tau$.

Due to competition between the two candidates, both offer the transfer that is most preferred by citizens: namely $\tau = 0$ if $\alpha \in (1,2]$ and $\tau = \int_i \pi_i \d i - K^G$ if $\alpha > 2$. This follows from a simple contradiction argument. If neither offer the specified transfer, then one can profitably deviate (by offering the specified transfer) and win the vote with probability one. If only one offers the specified transfer, the other can deviate (by offering the specified transfer) and win the vote with probability one half (rather than zero).
\end{proof}

As an obvious aside, it is worth noting that the government's utility is always increasing in the transfer it receives.

\begin{cor}
The government's payoff, $u_G$, is strictly increasing in the transfer, $\tau$, it receives. 
\end{cor}
\begin{proof}
Recall that the government's utility is:
\begin{align}
    u_G = \int_i \left[ - \alpha \pi_i (1 - p_i) - L_i - K_i \right] \d i + K^T,
\end{align}
where now $K^T = K^G + \tau + \int_{i \in \B} B_i \d i$. Then substituting in the equilibrium outcomes from \Cref{prop:eq with bribes}:
\begin{align}
    u_G = - \alpha \int_{i \in \B} 2 \pi_i \d i + K^T.
\end{align}
Now substituting in \Cref{eq:C3} and rearranging yields:
\begin{align}
    u_G &= \frac{1}{\chi} \left( \alpha - 1 + \frac{1}{\chi} \right) ( K^G + \tau) + \frac{1}{\chi} (1 - \alpha) \Pi \\
    \implies \frac{d u_G}{d \tau} &=  \frac{1}{\chi} \left( \alpha - 1 + \frac{1}{\chi} \right) > 0.
\end{align}
\end{proof}

\newpage
\section{Simulations}\label{sec:simulations}
This section presents simulations with a finite number of special interest groups. The goal here is to understand what happens when we relax the assumption that there is a mass of special interest groups, and that they are all arbitrarily small. The key finding here is that relaxing this assumption has minimal impact on the results.

In the proof to \Cref{prop:eq with bribes}, the continuum assumption allows us to ignore (i.e. treat as very small) the impact that an individual interest group has on the government's total budget. So taking one contribution from a single interest group does not affect the marginal value of political capital. Only contributions from a positive measure of interest groups will do so.\footnote{The assumption shows up in the proof to \Cref{prop:eq with bribes} in the claim that ``Spending these resources has a marginal benefit $\frac{d u_G}{d K_i}$'', and that this partial derivative is not treated as a function of $B_i$.} 
When we relax this assumption, taking contributions becomes ``lumpy'' in the sense that a single contribution can have a material impact on the government's overall budget of political capital.

These simulations are not designed to provide an extensive search of the parameter space. They are merely designed to provide reassurance that the mass of interest groups assumption we made in the main paper is not driving the insights of the model. To present additional challenge to our model (and to simplify the algorithm) we only allow the government to choose whether or not to do a deal with each special interest group in sequence. This presents two constraints on government behaviour compared to the model in \Cref{sec:model}. First, if the government chooses not to do a deal with interest group $i = n$, then it cannot do a deal with any group $i > n$ (where $i$ is the interest group's identity, and $n$ is an integer). Second, the government cannot change the interest groups it does deal with. Once a deal is done, it cannot be undone. Additionally, we assume that for each deal, the special interest group pays the maximum amount it is willing do, given the government's position at the point the deal is negotiated. Again, this helps simplify the implementation.

Given these additional constraints, we should expect these simulations to perform \emph{worse} (i.e. further away from the analytic results) than a more complex environment where the government does not face these constraints. Nevertheless, the simulated equilibria are close to the analytic predictions.

\paragraph{A simple example.} Set $\alpha = 3$, $\pi \sim U[1,11]$, $K^G = 10$, and $|I| = 10$ (i.e. only 10 special interest groups). %\footnote{In this particular simulation, the random draws of $\pi$ were: 5.17022, 8.20324, 1.00114, 4.02333, 2.46756, 1.92339, 2.8626, 4.45561, 4.96767, 6.38817} 
The small number of interest groups provides a more stringent test of the assumption. \Cref{fig:simulation_example} clearly shows that the general insight of the model is not lost when moving to a small number of special interest groups. However, it is difficult to see the exact value of $K_i$ and $B_i$ from the image. Exact results are listed in \Cref{tab:simulation_example}. 

\begin{figure}[ht!]
    \centering
    \includegraphics[scale=0.9]{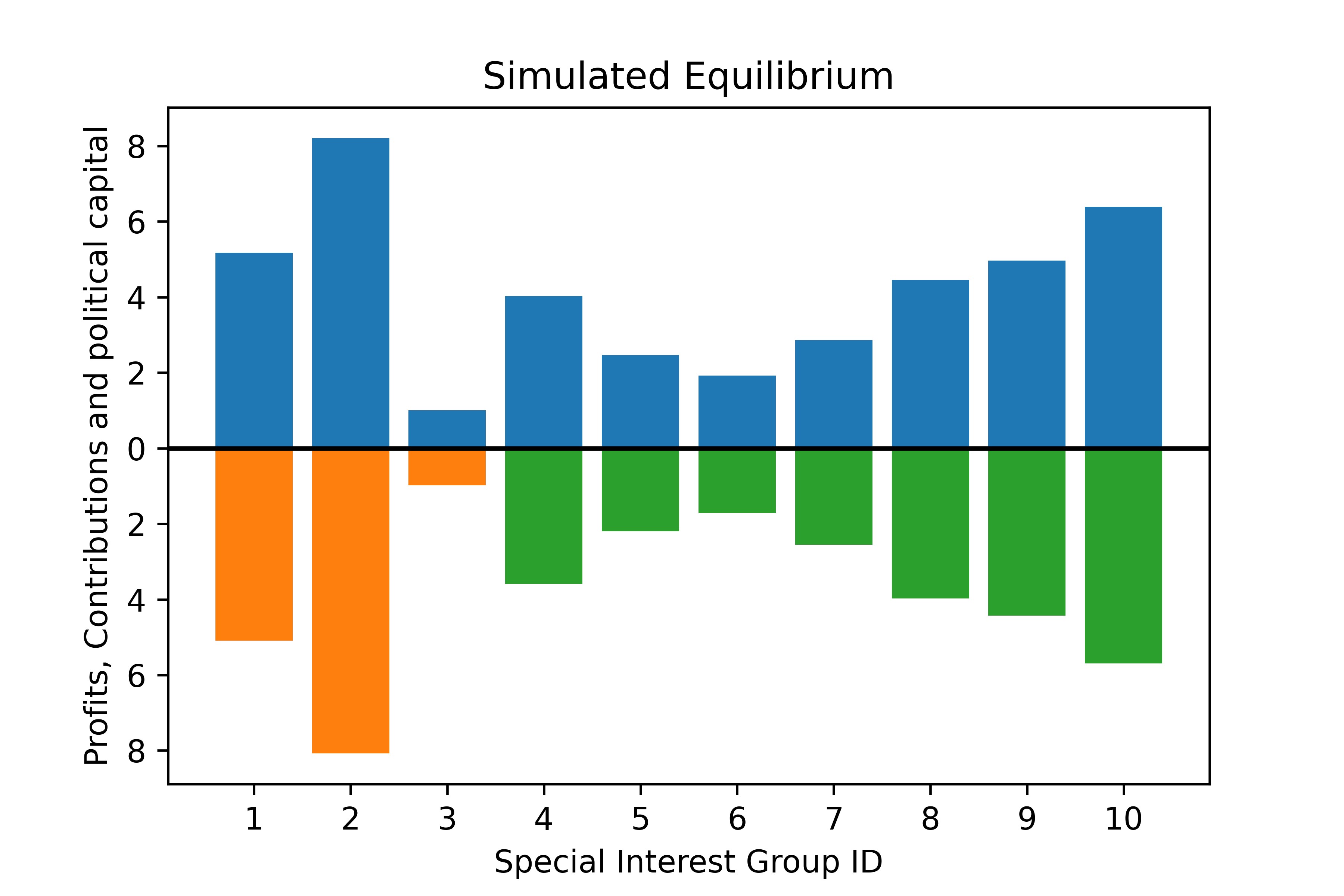}
    \caption{Simulated outcomes for $|I| = 10$, $\alpha = 3$, $K^G = 10$, $\pi \sim U[1,11]$. Special interest group benefits from bad law reported (in blue) above the centre line. Contributions (in orange) and political capital (in green) spending reported below the centre line.}
    \label{fig:simulation_example}
\end{figure}

\begin{table}[htbp!]
\setlength{\tabcolsep}{24pt}
\begin{center}
\caption{Exact results for the simulation reported in \Cref{fig:simulation_example}.}
\label{tab:simulation_example}
\begin{tabular}{llllll}
ID & $\pi_i$ & $B_i$    & $K_i$   & $B_i / \pi_i$  & $K_i^* / \pi_i$   \\
1  & 5.170 & 5.087 & 0     & 0.984 & -     \\
2  & 8.203 & 8.071 & 0     & 0.984 & -     \\
3  & 1.001 & 0.985 & 0     & 0.984 & -     \\
4  & 4.023 & 0     & 3.586 & -     & 0.891 \\
5  & 2.468 & 0     & 2.199 & -     & 0.891 \\
6  & 1.923 & 0     & 1.714 & -     & 0.891 \\
7  & 2.863 & 0     & 2.551 & -     & 0.891 \\
8  & 4.456 & 0     & 3.971 & -     & 0.891 \\
9  & 4.968 & 0     & 4.427 & -     & 0.891 \\
10 & 6.388 & 0     & 5.694 & -     & 0.891
\end{tabular}
\end{center}
\end{table}

\paragraph{A fuller test.} More substantively, we examine the equilibria as the number of special interest groups changes. Here, we simply keep track of the simulated equilibrium value of $Z$. Since we know that $K_i^* = \pi_i Z$ and $B_i^* = \pi_i (2 Z^{0.5} - Z)$, this is sufficient. Recall that 
\begin{align*}
    Z \equiv \frac{ K^G + \int_{i \in \B} \pi_i \d i }{ \int_{j \in \B^c} \pi_j \d j } ,
\end{align*}
and the key characteristic of equilibrium is that $Z = 1$. That is, the government has just enough capital to win all remaining contests with certainty. 

In order to keep the comparisons across different values of $|I|$ fair, we scale $K^G$ linearly with $|I|$. For simplicity (and to ensure that $K^G \geq \int_i \pi_i \d i$ is not possible) we choose $K^G = 1 \times |I|$. As before, $\alpha = 3$ and $\pi \sim U[1,11]$. We consider $|I| \in \{10,...,500\}$.

First, \Cref{fig:simulations_change_n} shows results for a single run (for each value of $|I|$). Next \Cref{fig:simulations_change_n_avg} shows the average values from 100 runs (for each value of $|I|$). Clearly there is more variability in a single run than an average, but there is otherwise little material difference.

\clearpage
\begin{figure}[!htb]
    \centering
    \vspace{-20mm}
    \includegraphics[scale=0.95]{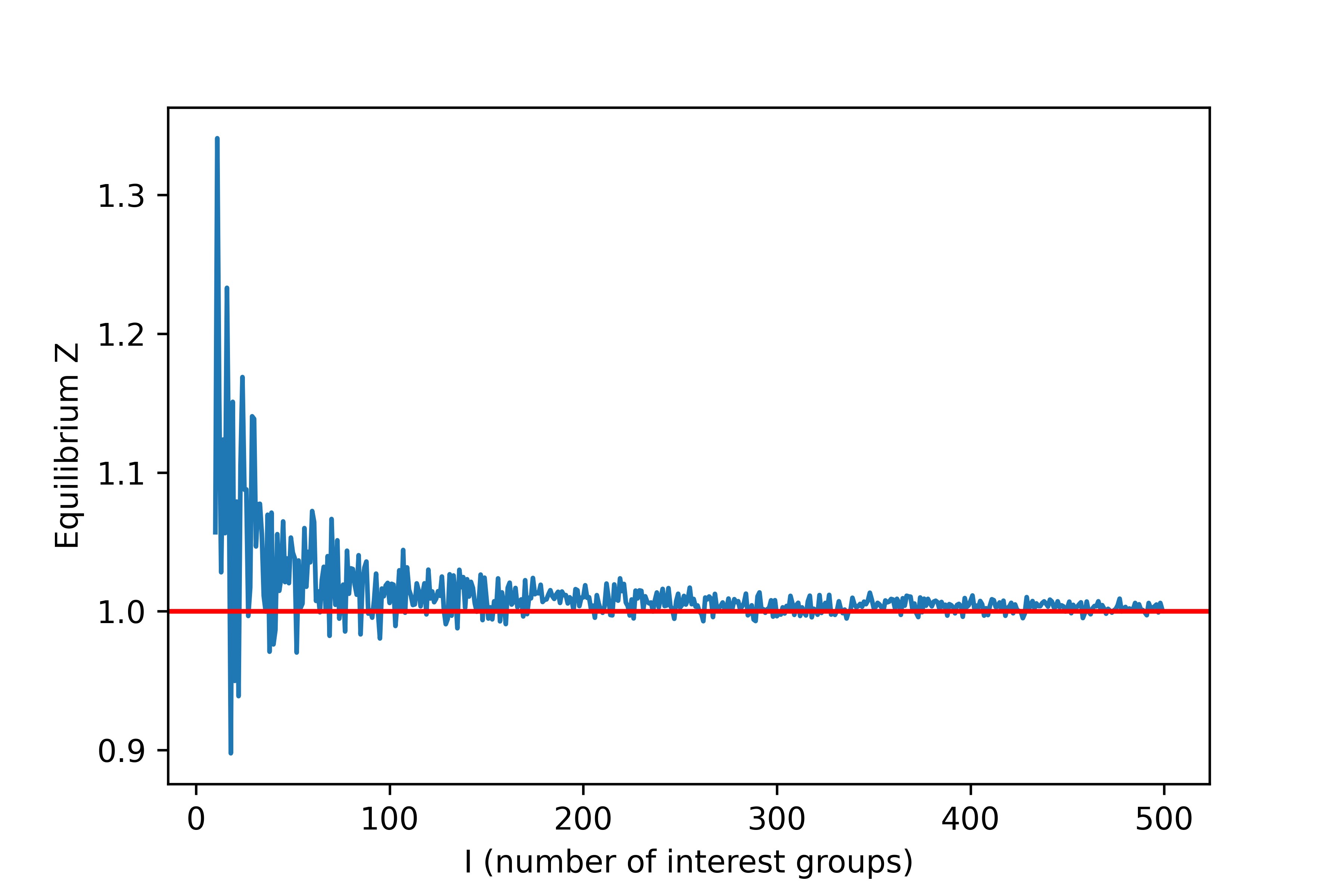}
    \caption{Simulated value for $Z$, for $|I| \in \{10,...,500\}$. Other parameters: $\alpha = 3$, $K^G = 1 \times |I|$, $\pi \sim U[1,11]$.}
    \label{fig:simulations_change_n}
\end{figure}
%\vspace{25mm}
\begin{figure}[!htb]
    \centering
    \hspace{-12.5mm}
    \includegraphics[scale=0.95]{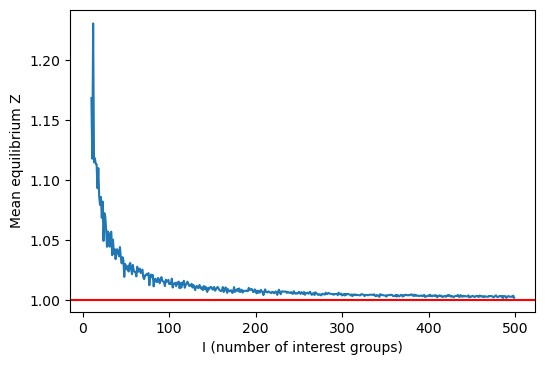}
    \caption{Mean simulated value for $Z$, for $|I| \in \{10,...,500\}$. Other parameters: $\alpha = 3$, $K^G = 1 \times |I|$, $\pi \sim U[1,11]$. Average of 100 runs (i.e. draws of $\pi$).}
    \label{fig:simulations_change_n_avg}
\end{figure}

\clearpage
\section{Technical Discussion of Extensions}\label{sec:extensions}
This appendix formally examines a range of alterations to the model presented in \Cref{sec:model}, and acts as a supplement to \Cref{sec:extensions discussion}.

\subsection{Timing: simultaneous choice of lobbying and political capital}
Here, the model is identical to \Cref{sec:model} with the following exception: stages two and three in the game happen simultaneously.

This has no impact on special interest groups' optimal decision rules in the ``third'' stage. They still take the government's decisions as given when choosing how much lobbying to do.
Hence, $L_i^* = (\pi_i K_i)^{0.5} - K_i$ if $K_i \leq \pi_i$, and $L_i^* = 0$ otherwise. The government, however, faces a different problem. While its utility function is unchanged, it must take interest groups' lobbying as given. This means it cannot take into account how its own choice of $K_i$ will affect $L_i$. This does not change relative allocation of capital across contests. However, it does affect the point at which the government chooses to stop spending more capital. 

Additional spending on a contest can deter lobbying by the interest group. In this extension, the government does not account for this ``deterrence effect'' when making its decisions. This lowers the marginal benefit of spending political capital in a contest. It therefore causes the government to stop spending additional capital sooner. Consequently, interest groups will still do some lobbying in the corner case when the government's has enough political capital to win all (remaining) contests with certainty.

\begin{lem}\label{lem:1 simultaneous}
Fix a set of contributions $B_i$ (and hence $\B$ and $K^T$). Then there is a unique equilibrium:

\noindent (A) \ If $B_i \neq 0$: $L_i^* = 0$, and $K_i^* = 0$. \\
\noindent (B) \ If $B_i = 0$: $L_i^* = (\pi_i K_i^*)^{0.5} - K_i^*$ if $K_i^* \leq \pi_i$ and $L_i^* = 0$ otherwise, and $K_i^* = \pi_i Z$ if $Z \leq \left( \frac{\alpha}{\alpha + 1} \right)^2$ and $K_i^* = \left( \frac{\alpha}{\alpha + 1} \right)^2 \pi_i$ otherwise,
where $Z = \frac{K^T}{\int_{j \in \B^c} \pi_j \d j}$
\end{lem}
\begin{proof}
Set up a standard Lagrangian problem for the government and take First Order Conditions as normal.
\begin{align}
    \mathcal{L} &= \int_{i \in \B^c} \left[ - \alpha (1 - p_i) \pi_i - L_i - K_i \right] \d i - \mu \left(\int_{i \in B^c} K_i \d i - K^T \right) \\
    \frac{d \mathcal{L}}{d K_i} &= \frac{\alpha L_i \pi_i}{(K_i + L_i)^2} - 1 - \mu.
\end{align}
by setting equal to zero and rearranging terms:
\begin{align}\label{eq:opt k simultaneous}
    K_i^* = \left( \frac{\alpha \pi_i L_i }{1 + \mu} \right)^{0.5} - L_i \text{  for all  } i \in \B^c
\end{align}
where $\mu \geq 0$ is the shadow value of relaxing the government's budget constraint. By the same logic as in Lemma 1, the government will never choose $K_i > \pi_i$ for any $i$. Therefore, we can focus on the first part of an interest group's decision rule. This gives two equations in two unknowns: $L_i^* = (\pi_i K_i)^{0.5} - K_i$ and \cref{eq:opt k simultaneous}.
Substituting one into the other and rearranging terms yields:\footnote{During this rearranging, $K_i^* = 0$ appears to be a solution. However, this cannot be an equilibrium because when $L_i = 0$, the marginal return to $K_i$ is infinite. This happens because the contest function is not differentiable at the point $K_i + L_i = 0$ (even though it is differentiable when one of $K_i$ and $L_i$ are zero).}

\begin{align}
    K_i^* = \pi_i   \left( \frac{\alpha}{\alpha + 1 + \mu} \right)^2 \quad \text{ and } \quad
    L_i^* = \pi_i \frac{\alpha (1 + \mu)}{(\alpha + 1 + \mu)^2}.
\end{align}
When the government's budget constraint binds then $\mu > 0$ and $\int_{i \in \B^c} K_i^* \d i = K^T$. Substituting in the value of $K_i^*$:
\begin{align}
    K^T &= \int_{i \in \B^c} \pi_i  \left( \frac{\alpha}{\alpha + 1 + \mu} \right)^2 \d i \\
    \frac{\alpha}{\alpha + 1 + \mu} &= \left( \frac{K^T}{\int_{i \in \B^c} \pi_i \d i} \right)^{0.5} \\
    \mu &= \alpha \left( \frac{K^T}{\int_{i \in \B^c} \pi_i \d i} \right)^{-0.5} - 1 - \alpha
\end{align}
Substituting this back into the equations for $K_i^*$ and $L_i^*$ (and rearranging) yields:
% CALCULATIONS FOR THE EQUATION BELOW
%\begin{align}
    %K_i^* = \pi_i   \left( \frac{\alpha - 1}{\alpha + (\alpha - 1) \left( \frac{K^T}{\sum_{i \in \B^c} \pi_i} \right)^{-0.5} - \alpha} \right)^2 \\
    %K_i^* = \pi_i   \left( \frac{\alpha - 1}{(\alpha - 1) \left( \frac{K^T}{\sum_{i \in \B^c} \pi_i} \right)^{-0.5}} \right)^2 \\
    %K_i^* = \pi_i   \left( \frac{1}{\left( \frac{K^T}{\sum_{i \in \B^c} \pi_i} \right)^{-0.5}} \right)^2 \\
    %K_i^* =  \frac{\pi_i K^T}{\int_{i \in \B^c} \pi_i \d i}
%\end{align}
%\begin{align}
    %L_i^* = \pi_i \frac{(\alpha - 1) (1 + (\alpha - 1) \left( \frac{K^T}{\sum_{i \in \B^c} \pi_i} \right)^{-0.5} - \alpha)}{(\alpha + (\alpha - 1) \left( \frac{K^T}{\sum_{i \in \B^c} \pi_i} \right)^{-0.5} - \alpha)^2} \\
    %L_i^* = \pi_i \frac{(\alpha - 1)^2 \left[ \left( \frac{K^T}{\sum_{i \in \B^c} \pi_i} \right)^{-0.5} - 1 \right]}{(\alpha - 1)^2 \left( \frac{K^T}{\sum_{i \in \B^c} \pi_i} \right)^{-1} } \\
    %L_i^* = \pi_i \left[ \left( \frac{K^T}{\int_{i \in \B^c} \pi_i \d i} \right)^{0.5} - \frac{K^T}{\int_{i \in \B^c} \pi_i \d i}  \right]
%\end{align}
\begin{align}
    K_i^* =  \frac{\pi_i K^T}{\int_{i \in \B^c} \pi_i \d i} \quad \text{ and } \quad L_i^* = \pi_i \left[ \left( \frac{K^T}{\int_{i \in \B^c} \pi_i \d i} \right)^{0.5} - \frac{K^T}{\int_{i \in \B^c} \pi_i \d i}  \right]
\end{align}
When the government's budget constraint does not bind then $\mu = 0$, and $\int_{i \in \B^c} K_i^* \d i < K^T$. Clearly this yields:
\begin{align}
    K_i^* = \pi_i   \left( \frac{\alpha}{\alpha + 1} \right)^2 \quad \text{ and } \quad
    L_i^* = \pi_i \frac{\alpha}{(\alpha + 1)^2}.
\end{align}
In this case the total government spending is $ \left( \frac{\alpha}{\alpha + 1} \right)^2 \int_{i \in \B^c} \pi_i$, so the budget constraint is slack whenever $K^T$ is larger than this term.
\end{proof}

\vspace{5mm}
This feeds through to negotiations over contributions. In this variant, the government will stop spending political capital in contest before the point where it wins for certain, the largest contribution it could extract from a special interest group is lower. This is because the special interest group's outside option (of fighting) is relatively more attractive -- so it is less willing to pay to avoid a fight.
Further, it also reduces the total amount of contributions the government will take. This is because taking contributions is less attractive -- as it now extracts smaller contributions and because the government needs less capital before it reaches the point where it no longer wants to spend on contests.  

\begin{rem}
Suppose $\alpha > 2$. Then the government takes fewer contributions when the government and special interest groups choose contest spending simultaneously than when the government chooses first.
\end{rem}
\begin{proof}
First, recall that in the sequential move version in the main paper, in equilibrium $K^G + \int_{i \in \B} B_i \d i = \int_{i \in \B^c} \pi_i \d i$, and $B_i = \pi_i$ for all $i \in B$. Therefore, $\int_{i \in \B} B_i \d i = \frac{1}{2} [\Pi - K^G]$.\footnote{To see this, use the fact that $B_i = \pi_i$ for all $i \in \B$ to get $K^G + \int_{i \in \B} \pi_i = \int_{i \in \B^c} \pi_i $. Then let $\Pi = \int_i \pi_i = \int_{i \in \B} \pi_i \d i + \int_{i \in \B^c} \pi_i \d i$. Using this definition, and some rearranging yields the result.}

Next, notice that $b_i^{max} = \pi_i (2 Z^{0.5} - Z)$ if $Z \leq \left( \frac{\alpha}{\alpha + 1} \right)^2$, and $b_i^{max} = \frac{\alpha}{\alpha + 1} \frac{\alpha + 2}{\alpha + 1} \pi_i$ otherwise. It follows directly from Lemma 2 that $b_i^{max} = p_i^* \pi_i + L_i^*$ and $b_i^{min} = \frac{1}{1 + \mu} (p_i \alpha \pi_i - L_i^* - (1+\mu)K_i^*)$. These equations do not depend on the particulars of behaviour in the second and third stages. Then using behaviour from \Cref{lem:1 simultaneous} yields characterisation of $b_i^{max}$. %: $b_i^{max} = \pi_i (2 Z^{0.5} - Z)$ if $Z \leq \left( \frac{\alpha}{\alpha + 1} \right)^2$ and $b_i^{max} = \frac{\alpha (\alpha + 2)}{(\alpha + 1)^2} \pi_i$ if $Z > \left( \frac{\alpha}{\alpha + 1} \right)^2$.
Now notice that $b_i^{max}$ is increasing in $Z$. So the most an interest group would pay is $\frac{\alpha}{\alpha + 1} \frac{\alpha + 2}{\alpha + 1} \pi_i$.

So, when the government (i) receives the largest contribution possible from interest groups, and (ii) takes contributions up to the point where it would no longer spend extra political capital on contests, we have:
\begin{align}
    K^G + \int_{i \in \B} \frac{\alpha}{\alpha + 1} \frac{\alpha + 2}{\alpha + 1} \pi_i \d i = \int_{i \in \B^c} \left( \frac{\alpha}{\alpha + 1} \right)^2 \pi_i \d i %\\
    %K^G + \frac{2 \alpha^2 + 2 \alpha}{(\alpha + 1)^2} \int_{i \in \B} \pi_i = \left( \frac{\alpha}{\alpha + 1} \right)^2 \Pi \\
\end{align}
Using $\Pi = \int_i \pi_i \d i$ and rearranging terms yields:
\begin{align}
    \int_{i \in \B} \pi_i \d i = \frac{1}{2} \left[ \frac{\alpha}{\alpha + 1} \Pi - \frac{\alpha + 1}{\alpha} K^G \right]
\end{align}
It is clear that for any finite $\alpha$, this is smaller than $\frac{1}{2} [\Pi - K^G]$, which was the quantity of contributions taken in the sequential-move version of the model. The fact that in the simultaneous-move version, we have $B_i < \pi_i$ for all $i \in \B$ then completes the proof.
\end{proof}

\subsection{Government: non-fungible capital}
Here, the model is identical to \Cref{sec:model} with the following exception: political capital spending has no opportunity cost, so
\begin{align}
    u_G = \int_i \left[ - \alpha \pi_i (1 - p_i) - L_i \right] \d i.
\end{align}
This creates two changes. First, the marginal benefit of spending political capital is now higher by one unit when the probability of winning a contest is not yet one. Second, there is now no benefit to doing deals if the political capital cannot \emph{usefully} be spent on winning contests.

\begin{prop}
\noindent If $\alpha > 1$ and $K^G \geq \int_i \pi_i \d i$, there is a unique equilibrium: $B_i = 0$, $K_i = \pi_i$, $L_i = 0$ for all $i \in I$.

\noindent If $\alpha > 1$ and $K^G < \int_i \pi_i \d i$, then in all equilibria: 

    (1) $B_i = \pi_i$, $K_i = 0$, $L_i = 0$ for all $i \in \B$ 
    
    (2) $B_i = 0$, $K_i = \pi_i$, $L_i = 0$ for all $i \in \B^c$
    
    (3) $\B$ is such that $\int_{j \in \B^c} \pi_j \d j = K^G + \int_{j \in \B} \pi_j \d j$.
\end{prop}

\begin{proof}
Lemma 1 holds unchanged. The only small change in the proof is that the new utility function leads to 
$$\frac{d \mathcal{L}}{d K_i} = \frac{1}{2} (\alpha - 1) \left( \frac{\pi_i}{K_i} \right)^{0.5} + 1 - \mu.$$

The only change to Lemma 2 is that $b_i^{min}$ does not exist when $Z > 1$. The government is \emph{never} willing to do a deal in this situation. This is because doing a deal involves a cost, but by assumption cannot bring a benefit.

This means that Case (1) of Step 2 in the proof to Proposition 1 no longer applies. As soon as $Z \geq 1$, the government is unwilling to do additional deals. Cases (2) and (3) then apply whenever $\alpha > 1$
\end{proof}

\subsection{Interest Groups: imperfect legislation}
Here, the model is identical to \Cref{sec:model} with the following exception: if a law $i$ is open, then the benefits to the interest group are $(1 - z_i) \pi_i$ and the costs to society are $\alpha (1 - z_i) \pi_i$, with $z_i \in [0,1]$.

Intuitively, this is a setting where the government is only able to mitigate, rather than completely remove, the harm a special interest group can impose on society. Further, its ability to mitigate differs across laws. The parameter $z_i$ captures the effectiveness of legislation, with $z_i = 1$ being the special case corresponding to the main model in \Cref{sec:model}. This is an intuitively appealing extension, as in many real-world settings legislation and regulation will only be partially effective. Preferences are then:

\begin{align*}
    u_i &= (1 - p_i) \pi_i + p_i (1 - z_i) \pi_i - (L_i + B_i)  \tag{2f} \\
    u_G &= \int_i [- \alpha \pi_i (1 - p_i) - \alpha p_i (1 - z_i) \pi_i - L_i - K_i] \d i  + K^T. \tag{3f}
\end{align*}

However, this change has no impact on the analytic results. We can merely replace the interest groups' benefits, $\pi_i$, with a weighted equivalent, where the weights capture the extent to which the government can mitigate harms regarding the relevant law. Define a weighted benefit $\pi_i^{'} \equiv z_i \pi_i$ and rearrange these utility functions to get:

\begin{align*}
    u_i &= (1 - p_i) \pi_i^{\prime} - (L_i + B_i) + C_i \tag{2f'} \\
    u_G &= \int_i [- \alpha (1 - p_i) \pi_i^{\prime} - L_i - K_i + C_i] \d i + K^T \tag{3f'}
\end{align*}
where $C_i = \pi_i (1 - z_i)$. Since $C_i$ is a (law specific) constant term and is unaffected by any agents' decisions, this problem is identical to that in the main model. The only difference is that we replace $\pi_i$ with $\pi^{\prime}_i$. Therefore all results in \Cref{sec:results} will be unaffected, save for the replacement of $\pi_i$ with $\pi_i^{\prime}$.

\newpage
\subsection{Interest Groups: lobbying after winning a contest}
Here, the model is identical to \Cref{sec:model} with the following exception: an interest group $i$ wins the contest over law $i$, then it must engage in a quantity of lobbying $\beta \pi_i$, with $\beta \in [0,1)$, in order to implement the bad law.\footnote{We assume this to be the case regardless of whether (a) both the interest group and the government spend resources contesting the law and the interest group wins, or (b) the government does not participate and the interest group wins the contest by participating but choosing $L_i = 0$.} 
Formally, this is a fourth stage of the model. At the end of the third stage, the winners of each contest are realised (see \Cref{sec:model}). After that, interest groups who won their contest and secured a `bad' law must do an amount of lobbying of at least $\beta \pi_i$, or else the law reverts to being `good'. Intuitively, this is a setting where the interest group needs to engage in some lobbying to ensure swift passage of the law, or to `pay off' members of the legislature.  

Conceptually, this `additional' lobbying is the same as the lobbying in the contests in \Cref{sec:model}. But it will be convenient to denote this additional lobbying $\Tilde{L}_i$ and to keep track of it separately from the lobbying in the contests. Doing so helps keep the derivations clear.

The impact of this variation on the model is simple. It reduces the benefit to the special interest group of winning a contest to $(1 - \beta) \pi_i$ (i.e. the original benefit of winning, less what must be paid to secure swift passage after winning). So, defining $\hat{\pi}_i = (1 - \beta) \pi_i$, the special interest groups' preferences are then:

\begin{align}
    u_i &= (1 - p_i) \hat{\pi}_i - (L_i + B_i).
\end{align}

When the government loses a contest (which happens with probability $(1 - p_i)$), it incurs a cost $\alpha \pi_i$ as before, and now also incurs a cost of $\beta \pi_i$ -- the lobbying the interest group does to secure swift passage of the bad law. So the government's preferences are then:

\begin{align}
    u_G &= \int_i [- (\alpha \pi_i + \Tilde{L}_i) (1 - p_i) - L_i - K_i] \d i  + K^T.
\end{align}

It is now convenient to define $\hat{\alpha} = \frac{\alpha + \beta}{1 - \beta}$. Then substituting in $\Tilde{L}_i = \beta \pi_i$ into the government's preferences and rearranging yields:

\begin{align}
    u_G &= \int_i [- \hat{\alpha} \hat{\pi}_i (1 - p_i) - L_i - K_i] \d i  + K^T.
\end{align}

We now have identical preferences as in \Cref{sec:model}, except with $\hat{\alpha}$, $\hat{\pi}_i$ rather than $\alpha$, $\pi_i$. The results in \Cref{sec:results} will then be unaffected, save for (a) the replacement of  $\alpha$, $\pi_i$ with $\hat{\alpha}$, $\hat{\pi}_i$, and (b) the additional lobbying when the interest group wins its contest. 

The key difference of this extension is that even when deals are possible, lobbying happens in equilibrium. This demonstrates that the finding of no equilibrium lobbying in the main text is not robust to even simple extensions.

\newpage
\section{Extra Results}\label{sec:extra_proofs}
\subsection{Non-cooperative bargaining}
Here, we show that nothing material changes if we have a non-cooperative bargaining stage, where agents play a simpler variant on the Nash Demand Game \cite{nash1953two}. We do need to assume that players do not play weakly dominated strategies in the `bargaining stage' \emph{and} that this behaviour is common knowledge. Otherwise, we could end up in the situation where, even though there exists a mutually agreeable deal, both agents make extreme offers. These could be mutual best responses because each agent is playing a best response due to the extreme demand of the other.

The model is the same is as in \Cref{sec:model}, except for the parts labelled `timing' and `formalising timing and strategies'. For clarity, we write these two parts anew.

\paragraph{Timing.} First, for each law $i$, the government and the special interest group $i$ simultaneously make demands $b^G_i \geq 0$ (the government) and $b_i \geq 0$ (the interest group). They do a \emph{deal} if and only if these demands are mutually compatible, in the sense that $b_i^G \leq b_i$. In that case, the special interest group makes a \emph{contribution}, $B_i$, that is somewhere between the two demands. We also assume that agents do not make demands that are weakly dominated strategies \emph{and} that this behaviour is common knowledge.
Second, the government first chooses an amount of political capital spending, $K_i \geq 0$, for each law $i$. Third, each special interest group chooses a quantity of lobbying $L_i \geq 0$. We can formalise the timing and agents' strategies as follows. 

\paragraph{Formalising timing and strategies.} There are two distinct stages. First, a `negotiation stage', where the government and interest groups agree \emph{deals}. Second, a `contest stage', where they contest laws.

In the `negotiation stage', for each law $i$, the government and the special interest group $i$ simultaneously make demands $d^G_i \geq 0$ (the government) and $d_i \geq 0$ (the interest group). Define a function $D_i: \mathbb{R}^2 \to \{0,1\}$ for each $i$, where $D_i(d^G_i,d_i) = \mathbf{1} \{ d^G_i \leq d_i \}$, and where $\mathbf{1}$ is the usual indicator function. A \emph{deal} is done if and only if $D_i = 1$. 
And define a \emph{contribution} function $B_i:\mathbb{R}^2 \times \{0,1\} \to \mathbb{R}$, where $B_i = f(d^G_i, d_i) \cdot D_i$, and $f:\mathbb{R}^2_+ \to \mathbb{R}$, where $f(d_i^G,d_i)$ is increasing in both arguments and $f(d_i^G,d_i) \in [d_i^G,d_i]$ for all $d_i^G,d_i$. Denote the set of special interest groups that do a deal $\mathcal{B} = \{i \in I: D_i = 1\}$, and its complement, $\mathcal{B}^c \equiv I \setminus \B$.

So the government's strategy in the negotiation stage is $b^G \in \mathbb{R}^I$, and each special interest group's strategy is $b_i \in \mathbb{R}$. Together, these strategies induce a set of deals done, $D: \mathbb{R}^I \times \mathbb{R}^{I} \to \{0,1\}^I$, and the sizes of the bribes paid, $B: \mathbb{R}^I \times \mathbb{R}^{I} \times \{0,1\}^I \to \mathbb{R}^I$.
These deals done plus bribes paid $(D,B) \in \{0,1\}^I \times \mathbb{R}^I$ are a sufficient statistic for the contest stage. Once the bribes have been agreed/paid, it does not matter what the strategies were that led to them.

The `contest stage' itself happens sequentially. First, the government then chooses an amount of political capital capital spending, $K_i \in \mathbb{R}^+$, for each contest $i$. Second, each special interest group $i$ chooses a quantity of lobbying $L_i \in \mathbb{R}^+$. 
So the government's strategy in the contest stage is $K: \{0,1\}^I \times \mathbb{R}^I \to \mathbb{R}^I$ (recalling that $(D,B) \in \{0,1\}^I \times \mathbb{R}^I$ is a sufficient statistic for the outcome of the negotiation), and each special interest group's strategy is $L_i: \{0,1\}^I \times \mathbb{R}^I \times \mathbb{R}^I \to \mathbb{R}$.

\subsubsection*{Impact on the solutions}
Moving to this non-cooperative bargaining set-up will have no meaningful impact on the results. However, it will introduce nuisance equilibrium multiplicity. Specifically, whenever the government and a special interest group $i$ \emph{do not} do a deal in equilibrium, any pair of demands $(d_i,d_i^G)$ such that $d_i^G > d_i$ will be part of an equilibrium. 

It suffices to show that whenever $b_i^{min} \leq b_i^{max}$, then $d^{G*}_i = d_i^* \in [b_i^{min} \leq b_i^{max}]$ and whenever. $b_i^{min} > b_i^{max}$, then $(d^G_i)^* > d_i^*$.
In other words, the non-cooperative bargaining leads to the same outcomes as the cooperative bargaining in \Cref{sec:model}. Whenever there exists a mutually agreeable deal, it gets done at some contribution that both parties find acceptable. And when there no mutually agreeable deal exists, no deal gets done.
We can formalise this claim in the following Lemma.

\begin{lem}
    (1) If $b_i^{min} \leq b_i^{max}$, then in all equilibria, $D_i^* = 1$ and $b_i^{G*} = b_i^* \in [b_i^{min} , b_i^{max}]$.

    (2) Otherwise, $D_i^* = 0$ and $(d_i^G)^* > b_i^*$, with $b_i^G > d_i^{max}$ and $b_i < d_i^{min}$.
\end{lem}

\begin{proof}
There are two parts: (1) $b_i^{max} \geq b_i^{min}$, and (2) $b_i^{max} < b_i^{min}$.

\paragraph{Part (1)} Show that $d_i^G = d_i$ is an equilibrium, then show that common knowledge that players do no play weakly dominated strategies implies that $d_i, d_i^G \in [b_i^{min}, b_i^{max}$], then show that there is a profitable deviation for any $d_i^G \neq d_i$.

\textbf{(i)} Suppose $d_i^G = d_i \in [b_i^{min} , b_i^{max}]$. If the government deviates to some $d_i^{G \prime} > d_i^G$, then there is no deal -- which is strictly worse. And if it deviates to some $d_i^{G \prime} < d_i^G$, then $B_i^{\prime} < B_i$ -- which is strictly worse. If the interest group deviates to some $d_i^{\prime} > d_i$, then $B_i^{\prime} > B_i$ -- which is strictly worse. And if it deviates to some $d_i^{\prime} < d_i$, then there is no deal -- which is strictly worse. Therefore $d_i^G = d_i \in [b_i^{min} , b_i^{max}]$ is an equilibrium.

\textbf{(ii)} Here we need the assumption that players do not play weakly dominated strategies. 
$d_i > b_i^{max}$ is weakly dominated by $d_i^{\prime} = b_i^{max}$. (a) if $d_i^G > d_i > d_i^{\prime}$, both lead to no deal (so no difference in payoff). (b) if $d_i \geq d_i^G > d_i^{\prime}$, then $d_i^{\prime}$ leads to no deal while $d_i$ leads to a deal that is strictly worse than no deal. (c) if $d_i > d_i^{\prime} \geq d_i^G$, then both lead to a deal, but $B_i^{\prime} < B_i$, so $d_i^{\prime}$ is strictly better than $d_i$. An analogous argument shows that $d_i^G < b_i^{min}$ is also weakly dominated. The idea here is the same as for simple second-price auctions: bidding more than your value is weakly dominated.

If $b_i^{min} \leq d_i^G < d_i \leq b_i^{max}$, then government has a profitable deviation to $d_i^{G \prime} = d_i$ -- as this leads to $B_i^{\prime} > B_i$, which is strictly preferred by the government.

Common knowledge that agents will not play weakly dominated strategies then implies that the special interest group will not play $d_i < b_i^{min}$, as doing so becomes weakly dominated conditional on knowing $d_i^G \geq b_i^{min}$. Similarly, this common knowledge implies $d_i^G \leq b_i^{max}$. So we are focusing on $d_i, d_i^G \in [b_i^{min}, b_i^{max}]$.

\textbf{(iii)} If $d_i > d_i^G$ then the government has a profitable deviation to $d_i^{G \prime} = d_i$ as this leads to a deal (and because we are assuming $d_i \in [b_i^{min}, b_i^{max}]$).

\paragraph{Part (2)} Show that failure of any of the three conditions leads to a profitable deviation. Then show that the three conditions together is an equilibrium.
\textbf{(i)} Suppose $d_i^G \leq d_i$. Then $D_i = 1$ and $B_i \in [d_i^G, d_i]$. (a) The government has a profitable deviation if $B_i < b_i^{min}$ and (b) the interest group has a profitable deviation if $B_i > b_i^{max}$. But as $b_i^{max} < b_i^{min}$, either (a) or (b) must hold for any value of $B_i$. Therefore at least one agent has a profitable deviation. 
\textbf{(ii)} Suppose $d_i^G \leq b_i^{max}$ and $d_i^G > d_i$. Then there is a profitable for the special interest group: $d_i^{\prime} = d_i^G$, which yields a deal at a value of $B_i$ that she prefers to no deal. So $b_i^{max} \geq d_i^G > d_i$ cannot be an equilibrium.
\textbf{(iii)} Suppose $d_i \geq b_i^{min}$ and $d_i^G > d_i$. Then there is a profitable for the government: $d_i^{G \prime} = d_i$, which yields a deal at a value of $B_i$ that he prefers to no deal.

\textbf{(iv)} Suppose $d_i^G > d_i$, with $d_i^G > b_i^{max}$ and $d_i < b_i^{min}$. A deviation for the government would either lead to no deal (i.e. no change in outcome) or a deal at $B_i < b_i^{min}$ -- which is strictly worse than no deal. Similarly, a deviation for the special interest group would either lead to no deal (i.e. no change in outcome) or a deal at $B_i > b_i^{max}$ -- which is a strictly worse outcome than no deal. So no profitable deviations. Hence an equilibrium.
\end{proof}

\subsection{Lagrange multipliers}
We now prove that with a measure of agents, the Lagrange multiplier does not change when the government does a single deal with interest group $i$. While obvious, it is asserted but not proved in Lemma 2.

%\subsection*{Lemma 3}
\begin{lem}
    $\hat{\mu} = \mu$
    %$\frac{\pi_i}{\Pi} \to 0 \ \forall i \ \text{ as } n \to \infty \implies \frac{\mu}{\hat{\mu}} \to 1 \ \text{ as } n \to \infty$
\end{lem}

%\subsection*{Lemma 3: Proof}
\begin{proof}
As with Lemma 1, it will be clearer to do the derivation for a finite number ($n$) of special interest groups, and then take the limit as $n \to \infty$.
From the proof to Lemma 1, we have that $\mu = 0.5 (\alpha - 1) \pi_i^{0.5} (K_i^*)^{-0.5}$. And then substituting in $K_i^* = \pi_i Z$ yields $\mu = 0.5 (\alpha - 1) Z^{-0.5}$. It is then clear that $\frac{\mu}{\hat{\mu}} = \left( \frac{\hat{Z}}{Z} \right)^{0.5}$. $Z$ is defined when the government does not do a deal with $i$, and $\hat{Z}$ is when it does. For convenience, we we $K^T$ here refer to the government's total political capital when it does not do a deal with $i$. So we have:
\begin{align*}
    Z = \frac{K^T}{\frac{1}{n} \left( \sum_{j \in \B^c \setminus i} \pi_j + \pi_i \right)} \quad \text{ and } \quad 
    \hat{Z} = \frac{K^T + \frac{1}{n} B_i}{\frac{1}{n} \sum_{j \in \B^c \setminus i} \pi_j }
\end{align*}
It is then immediate that as $n \to \infty$, we have (i) $K^T = K^T$, (ii) $\frac{1}{n} \sum_{j \in \B^c \setminus i} \pi_j \to \int_{j \in \B^c \setminus i} \pi_j \d j$, (iii) $\frac{1}{n} \pi_i \to 0$, and (iv) $\frac{1}{n} B_i \to 0$. Therefore we have
\begin{align*}
   \lim_{n \to \infty} Z = \frac{K^T}{\int_{j \in \B^c \setminus i} \pi_j} \quad \text{ and } \quad 
    \lim_{n \to \infty} \hat{Z} = \frac{K^T}{\int_{j \in \B^c \setminus i} \pi_j } 
\end{align*}
Which completes our result.
\end{proof}

\end{document}